\documentclass[conference]{IEEEtran}

\usepackage{ifthen}

\newboolean{noapp}
\setboolean{noapp}{false}
\newcommand{\ifshort}[2]{\ifthenelse{\boolean{noapp}}{#1}{#2}}

\newboolean{arxv}
\setboolean{arxv}{true}
\newcommand{\ifarxiv}[2]{\ifthenelse{\boolean{arxv}}{#1}{#2}}

\let\citep\cite
\usepackage{cite}
\IEEEoverridecommandlockouts

\usepackage[utf8]{inputenc}             %
\usepackage[T1]{fontenc}                %
\usepackage{url}                        %
\usepackage{booktabs}                   %
\usepackage{amsfonts}                   %
\usepackage{nicefrac}                   %
\usepackage{microtype}                  %
\usepackage[dvipsnames]{xcolor}         %
\usepackage{amsmath}
\usepackage{amssymb}
\usepackage{amsthm}
\usepackage{mathtools}
\usepackage{bm}
\usepackage{optidef}
\usepackage{comment}
\usepackage{calc}
\usepackage[inline]{enumitem}

\usepackage{subcaption}
\usepackage{hyperref}                   %
\ifarxiv{
	\hypersetup{
		colorlinks=true,
		urlcolor=RoyalBlue,
		citecolor=Green
	}
}{
	\hypersetup{
		draft,                          %
		colorlinks=true,
		urlcolor=RoyalBlue,
		citecolor=Green
	}
}

\usepackage{marginfix}  
\usepackage{sidenotes}  
\makeatletter
\RenewDocumentCommand\sidenotetext{oo+m}{
	\IfNoValueOrEmptyTF{#1}{%
		\@sidenotes@placemarginal{#2}{\textsuperscript{\{\thesidenote\}}{}~\scriptsize{#3}}%
		\refstepcounter{sidenote}%
	}{%
		\@sidenotes@placemarginal{#2}{\textsuperscript{#1}~#3}%
	}%
}
\RenewDocumentCommand\sidenotemark{o}{
	\@sidenotes@multichecker%
	\IfNoValueOrEmptyTF{#1}{%
		\@sidenotes@thesidenotemark{\{\thesidenote\}}%
	}{%
		\@sidenotes@thesidenotemark{#1}%
	}%
	\@sidenotes@multimarker%
}
\makeatother

\newcommand{\sm}{\setminus}

\newcommand{\T}{\mathsf{T}}
\newcommand{\given}{\,\vert\,}
\newcommand{\suff}{\succcurlyeq}
\newcommand{\ffus}{\preccurlyeq}
\newcommand{\ui}[2]{UI(M : {#1} \!\setminus\! {#2})}
\newcommand{\ri}[2]{RI(M : {#1} ; {#2})}
\newcommand{\si}[2]{SI(M : {#1} ; {#2})}

\newcommand{\df}[2]{\delta(M:{#1} \!\setminus\! {#2})}
\newcommand{\dfg}[2]{\delta_{G}(M: {#1} \!\setminus\! {#2})}
\newcommand{\dfh}[2]{\widehat{\delta}_G(M : {#1} \!\setminus\! {#2})}
\newcommand{\sigx}{\Sigma_{X|M}}
\newcommand{\sigy}{\Sigma_{Y|M}}
\newcommand{\sigm}{\Sigma_{M}}
\newcommand{\sigt}{\Sigma_{T}}
\newcommand{\sigth}{\widehat{\Sigma}_{T}}
\newcommand{\hx}{H_{X}}
\newcommand{\hy}{H_{Y}}
\newcommand{\sampm}{\mathsf{M}}  
\newcommand{\sampx}{\mathsf{X}}
\newcommand{\sampy}{\mathsf{Y}}
\newcommand{\xor}{\oplus}

\DeclareMathOperator{\Tr}{Tr}

\newcommand\indept{\protect\mathpalette{\protect\independenT}{\perp}}
\def\independenT#1#2{\mathrel{\rlap{$#1#2$}\mkern2mu{#1#2}}}

\makeatletter
\def\nsuff{\mathrel{\mathpalette\c@ncel\suff}}
\def\c@ncel#1#2{\m@th\ooalign{$\hfil#1\mkern1mu/\hfil$\crcr$#1#2$}}
\makeatother

\allowdisplaybreaks

\input{widebar.tex}

\makeatletter
\def\thm@space@setup{\thm@preskip=\parskip
\thm@postskip=0pt}
\makeatother
\newtheoremstyle{compact}%
{}            
{}            
{\itshape}    
{\parindent}  
{\bfseries}   
{.}           
{2mm}         
{}            
\newtheoremstyle{compactnonitalic}%
{}            
{}            
{}            
{\parindent}  
{\bfseries}   
{.}           
{2mm}         
{}            
\theoremstyle{compact}

\newtheorem{theorem}{Theorem}
\newtheorem{lemma}[theorem]{Lemma}   
\newtheorem{proposition}[theorem]{Proposition}
\newtheorem{corollary}[theorem]{Corollary}
\newtheorem{definition}{Definition}

\theoremstyle{compactnonitalic}
\newtheorem{remark}{Remark}

\newlength{\origabovedisplayskip}
\newlength{\origbelowdisplayskip}
\newlength{\origabovedisplayshortskip}
\newlength{\origbelowdisplayshortskip}
\expandafter\def\expandafter\normalsize\expandafter{%
    \normalsize
    \setlength\abovedisplayskip{3pt}
    \setlength\belowdisplayskip{3pt}
    \setlength\abovedisplayshortskip{3pt}
    \setlength\belowdisplayshortskip{3pt}
}

\newcommand\blfootnote[1]{%
  \begingroup
  \renewcommand\thefootnote{}\footnote{#1}%
  \addtocounter{footnote}{-1}%
  \endgroup
}

\title{Partial Information Decomposition via Deficiency\\for Multivariate Gaussians%
	\vspace{-1mm}%
}

\author{%
    \IEEEauthorblockN{Praveen Venkatesh*}
    \IEEEauthorblockA{%
	\href{mailto:praveen.venkatesh@alleninstitute.org}{\texttt{praveen.venkatesh@alleninstitute.org}} \\
	\textit{Allen Institute, Seattle, WA} \\
	\textit{University of Washington, Seattle, WA}%
	}
    \vspace{-8mm}
    \and
    \IEEEauthorblockN{Gabriel Schamberg*}
    \IEEEauthorblockA{%
	\href{mailto:gabes@mit.edu}{\texttt{gabes@mit.edu}} \\
	\textit{Picower Institute for Learning and Memory} \\
    \textit{Massachusetts Institute of Technology, Cambridge, MA}%
	}
    \vspace{-8mm}
}

\begin{document}

\ifarxiv{
\pagestyle{plain}
\thispagestyle{plain}
}{}

\maketitle

\begin{abstract}

	Bivariate partial information decompositions (PIDs) characterize how the information in a ``message'' random variable is decomposed between two ``constituent'' random variables in terms of \emph{unique}, \emph{redundant} and \emph{synergistic} information components.
	These components are a function of the joint distribution of the three variables, and are typically defined using an optimization over the space of all possible joint distributions.
	This makes it computationally challenging to compute PIDs in practice and restricts their use to low-dimensional random vectors.
	To ease this burden%
	, we consider the case of jointly Gaussian random vectors in this paper.
	This case was previously examined by Barrett~\cite{barrett2015exploration}, who showed that certain operationally well-motivated PIDs reduce to a \emph{closed form} expression for \emph{scalar messages}.
	Here, we show that Barrett's result does not extend to vector messages in general, and characterize the set of multivariate Gaussian distributions that reduce to closed-form.
	Then, for all \emph{other} multivariate Gaussian distributions, we propose a convex optimization framework for approximately computing a specific PID definition based on the statistical concept of \emph{deficiency}.
	Using simplifying assumptions specific to the Gaussian case, we provide an efficient algorithm to approximately compute the bivariate PID for multivariate Gaussian variables with tens or even hundreds of dimensions.
	We also theoretically and empirically justify the goodness of this approximation.
\end{abstract}

\section{Introduction}

Partial information decompositions (PIDs) provide a framework for characterizing the joint information content of three or more random variables.
The three-variable case is usually discussed in terms of how \emph{two} random variables, $X$ and $Y$, contribute to the information about a ``message'' variable, $M$, and is hence referred to as the \emph{bivariate} decomposition.
More precisely, a bivariate PID decomposes the mutual information $I\bigl(M; (X, Y)\bigr)$ into four additive components: the information about $M$ that is (i) unique to $X$; (ii) unique to $Y$; (iii) redundant in both $X$ and $Y$; and (iv) synergistic, i.e., cannot be obtained from $X$ or $Y$ \emph{individually}, but is present in their combination $(X, Y)$.
For an intuitive example, suppose $M = [M_1, M_2, M_3]$, $X = [M_1, M_2, M_3 \xor Z]$, and $Y = [M_2, Z]$, where $M_1, M_2, M_3, Z \sim \text{i.i.d.~Ber}(1/2)$ and $\xor$ represents bitwise-\textsc{xor}.
Then, $I\bigl(M ; (X, Y)\bigr)$ is 3~bits, and can be decomposed as follows: $X$ has 1~bit of unique information about $M$ (i.e., $M_1$) that is not contained in $Y$, while $Y$ has zero unique information.
$X$ and $Y$ have 1~bit of redundant information, i.e., $M_2$, which can be recovered from \emph{either} $X$ or $Y$.
Finally, $X$ and $Y$ have 1~bit of synergistic information, i.e., $M_3$, which is \emph{not} available in either $X$ or $Y$ \emph{individually}, but can only be decoded when $X$ and $Y$ are taken \emph{together}.
Multiple competing definitions have been proposed for unique, redundant and synergistic information~\cite{williams2010nonnegative, bertschinger2014quantifying, harder2013bivariate, banerjee2018unique, niu2019measure, finn2018pointwise, kay2018exact}, however their differences are only poorly understood (see~\cite{lizier2018information} for a review).
\blfootnote{*Equal contribution. The full version of this paper, including appendices, is available online~\cite{full_version}. Part of this work was done while P.~Venkatesh was a Ph.D.\ student at Carnegie Mellon University, Pittsburgh, PA.}

With the development of these formal definitions, PIDs have begun to find uses in a wide variety of contexts.
Many works, including ours, have used the PID to analyze biological and neural systems~\cite{colenbier2020disambiguating, boonstra2019information, krohova2019multiscale, timme2018tutorial, pica2017quantifying, gat1999synergy, schneidman2003synergy, brenner2000synergy, venkatesh2020information}, e.g., to understand how neural activity and responses jointly encode stimuli.
PID has also been used to characterize synergistic interactions in financial markets~\cite{scagliarini2020synergistic}, and to define new measures of bias in machine learning~\cite{dutta2020information}.
Unfortunately, we find that only a few of these definitions have clear operational interpretations, and \emph{these} definitions involve an optimization over the space of all joint distributions over $(M, X, Y)$, making them hard to compute in practice.
The absence of efficient methods to compute these PIDs, particularly for high dimensions, has been a bottleneck in their wider adoption.

Towards finding efficient ways of computing PIDs, this paper focuses on the bivariate PID for jointly Gaussian random vectors $M$, $X$ and $Y$.
Barrett~\cite{barrett2015exploration} previously examined this problem for the case of \emph{scalar messages} $M$, and showed that certain operationally well-motivated PID definitions can be reduced to a PID that has a \emph{closed form} expression in the Gaussian case (Section~\ref{sec_background}).
This has enabled the efficient computation of these operationally sound PIDs for Gaussians with scalar $M$.
However, the case of \emph{vector} $M$ has remained unresolved in the PID literature.

This paper has two main contributions concerning the bivariate PID for fully multivariate Gaussian random variables:%
\footnote{We use the term ``fully \emph{multivariate} Gaussian'' here to refer to instances of a \emph{bivariate} PID where $M$ may also be a \emph{vector} (in contrast to Barrett's work).
This is not to be confused with \emph{multivariate PIDs}, where information about $M$ is decomposed among more than two variables.}
\begin{enumerate*}[label=(\arabic*)]
	\item In Section~\ref{sec_mult_gaussian}, we show that for \emph{vector} $M$, Barrett's result (and hence the reduction to closed-form) does not apply in general.
		We introduce a key concept underpinning Barrett's result called \emph{Blackwell sufficiency}~\cite{blackwell1953equivalent}, using which we characterize the condition under which Barrett's result extends to fully multivariate Gaussians.
	\item %
		For cases where this condition \emph{does not} hold, we provide a convex optimization framework to efficiently approximate a PID definition that is based on the statistical concept of \emph{deficiency}~\cite{banerjee2018computing} (Section~\ref{sec_def_approx}).
		We also provide results showing how this approximation bounds the true Gaussian PID and specific cases when it is guaranteed to agree.
		In Section~\ref{sec_simulations}, we present numerical experiments that validate and go beyond our analytical results.
\end{enumerate*}
Finally, we conclude with a discussion of open questions in Section~\ref{sec_discussion}.

\section{Introduction to PID and Barrett's Result}
\label{sec_background}

Let $M$, $X$ and $Y$ be random variables with sample spaces $\sampm$, $\sampx$ and $\sampy$ respectively, and joint density $P_{MXY}$.
Then, a bivariate Partial Information Decomposition is loosely defined as a set of non-negative functions $UI(M{\,:\,}X{\,\sm\,}Y)$, $UI(M{\,:\,}Y{\,\sm\,}X)$, $RI(M{\,:\,}X;Y)$ and $SI(M{\,:\,}X;Y)$ that satisfy:
\begin{align} \label{eq_pid}
	I\bigl(M; (X,Y)\bigr) &= \ui{X}{Y} + \ui{Y}{X} \\
						  &\hphantom{= \vphantom{UI}}+ \ri{X}{Y} + \si{X}{Y}, \notag \\
    I(M; X) &= \ui{X}{Y} + \ri{X}{Y} \label{eq_pid_x}, \\
    I(M; Y) &= \ui{Y}{X} + \ri{X}{Y} \label{eq_pid_y}.
\end{align}
Here, $UI(M{\,:\;}X{\,\sm\,}Y)$ represents the information about $M$ uniquely%
\footnote{Note that the $X \setminus Y$ in the argument of $UI$ is purely notational, and does not represent a set difference. All quantities in~\eqref{eq_pid} are functions of $P_{MXY}$.}
present in $X$ and not in $Y$, while $RI(M{\,:\;}X;Y)$ and $SI(M{\,:\;}X ; Y)$ represent the redundant and synergistic information about $M$ contained between $X$ and $Y$.
Usually, we also require that $RI$ and $SI$ are symmetric in $X$ and $Y$ (e.g., see axioms in~\citep{williams2010nonnegative, bertschinger2014quantifying}).

For brevity, we may also write the terms in the RHS of~\eqref{eq_pid} as $UI_X$, $UI_Y$, $RI$, and $SI$ respectively.
Since we have three equations and four unknowns, defining any one these four unknowns suffices to determine the other three.

We now state two PID definitions (which are non-identical, in general) in order to present the main result of Barrett~\cite{barrett2015exploration}.
Their interpretations are given later in Section~\ref{sec_mult_gaussian}.
\begin{definition}[MMI-PID \cite{barrett2015exploration}] \label{def_mmi_pid}
	Let the redundant information about $M$ present in both $X$ and $Y$ be the \textbf{M}inimum of their respective \textbf{M}utual \textbf{I}nformations with $M$ (hence ``MMI''):
	\begin{equation} \label{eq_mmi_pid}
		RI_{\textup{MMI}}(M : X ; Y) \coloneqq \min\{I(M ; X),I(M ; Y)\}.
	\end{equation}
\end{definition}
\noindent Along with equations \eqref{eq_pid}--\eqref{eq_pid_y}, this fully determines the MMI-PID, i.e., $UI_{\text{MMI}}(M {:\,} X {\setminus} Y)$, $UI_{\text{MMI}}(M {:\,} Y {\setminus} X)$, and $SI_{\text{MMI}}(M {:\,} X ; Y)$ are now well-defined.%
\begin{remark} \label{rem_gauss_mmi_pid}
	When $P_{MXY}$ is jointly Gaussian, the MMI-PID can be written in closed form using the following well-known identity (derived in Appendix~\ref{app_misc_derivs} for completeness\ifarxiv{}{~\cite{full_version}}):
	\begin{equation} \label{eq_gauss_mmi_pid}
		I(M ; Z) = \frac{1}{2} \log\biggl(\frac{\mathrm{det}(\Sigma_M)}{\mathrm{det}(\Sigma_M - \Sigma_{MZ}^{}\Sigma_Z^{-1}\Sigma_{MZ}^\T)}\biggr),
	\end{equation}
	where $Z$ could be $X$, $Y$, or their concatenation $[X^\T, Y^\T]^\T$, $\Sigma_M$ and $\Sigma_Z$ are the covariance matrices of $M$ and $Z$, and $\Sigma_{MZ}$ is their cross-covariance matrix.
	The closed-form expression follows by using \eqref{eq_gauss_mmi_pid} in \eqref{eq_pid}--\eqref{eq_mmi_pid}.
\end{remark}
\begin{definition}[$\sim$-PID\footnote{This PID is also sometimes referred to as the BROJA PID in the literature.} \cite{bertschinger2014quantifying}] \label{def_tilde_pid}
	The unique information about $M$ present in $X$ and not in $Y$ is given by
	\begin{equation}
		\widetilde{UI}(M : X \setminus Y) \coloneqq \min_{Q \in \Delta_P} I_Q(M ; X \given Y),
	\end{equation}
	where $\Delta_P \coloneqq \{Q_{MXY}: Q_{MX} = P_{MX},\; Q_{MY} = P_{MY}\}$ and $I_Q$ is the conditional mutual information over the joint distribution $Q_{MXY}$.
\end{definition}
\begin{theorem}[Barrett~\cite{barrett2015exploration}] \label{thm_barrett}
	If $M$, $X$ and $Y$ are jointly Gaussian and $M$ is scalar, then the $\sim$-PID of Definition~\ref{def_tilde_pid} reduces to the (simpler) MMI-PID of Definition~\ref{def_mmi_pid}.%
	\footnote{Barrett's original theorem statement is actually slightly more general: all PIDs satisfying assumption $(*)$ of Bertschinger et al.~\cite{bertschinger2014quantifying} reduce to the MMI-PID for jointly Gaussian $P_{MXY}$ and scalar $M$. This includes a few other well-known definitions such as \cite{williams2010nonnegative} and \cite{harder2013bivariate}.}
\end{theorem}
\begin{remark}
	Theorem~\ref{thm_barrett} is significant because the $\sim$-PID is operationally well-motivated (as we will see), while the MMI-PID can be expressed in closed form for jointly Gaussian $P_{MXY}$ as shown in Remark~\ref{rem_gauss_mmi_pid}.
	This equivalence has therefore inspired usage of the MMI-PID in several studies that apply PIDs to real-world data~\cite{scagliarini2020synergistic, colenbier2020disambiguating, boonstra2019information, krohova2019multiscale}.
\end{remark}

\section{Main Results}

First we show that Barrett's result (Theorem~\ref{thm_barrett}) does not extend to vector $M$ in general, so the closed-form expression from Remark~\ref{rem_gauss_mmi_pid} may not always apply.
\vspace{-6pt}
\begin{proof}[\indent Counterexample]
	Suppose $M = [M_1, M_2]$, $X = M_1 + Z_1$ and $Y = M_2 + Z_2$, with $M_1$, $M_2$, $Z_1$, $Z_2 \sim \text{ i.i.d.\ } \mathcal N(0, 1)$.
	Based on the intuition from the introduction, $X$ and $Y$ have equal amounts of \emph{unique} information about $M$, and one can show that this holds for the $\sim$-PID.
	However, since $I(M ; X) = I(M ; Y)$, $X$ and $Y$ have only \emph{redundant} information under the MMI-PID, and no unique information.
	Thus, the $\sim$-PID does not always reduce to the MMI-PID for vector $M$.
\end{proof}
\vspace{-6pt}
This counterexample is shown more formally in Appendix~\ref{app_misc_derivs}.

\subsection{When do Multivariate Gaussians have Closed-form PIDs?}
\label{sec_mult_gaussian}

To answer this question, we make some novel and important observations about Barrett's result by breaking it down into a few distinct steps:
\begin{enumerate}[label=\Alph*), leftmargin=*, nosep]
	\item Given \emph{any} PID definition that satisfies the basic equations \eqref{eq_pid}--\eqref{eq_pid_y}, if $P_{MXY}$ is such that either $UI_X = 0$ or $UI_Y = 0$, then that PID reduces to the MMI-PID.
	\item Barrett considers a few specific PID definitions (including the $\sim$-PID mentioned above).
		In fact, we can consider all PIDs that satisfy the following property: $UI_Y = 0$ if and only if $X$ is \emph{Blackwell sufficient} for $M$ with respect to $Y$ (formally defined below).
		We call such PIDs \emph{Blackwellian}.%
		\footnote{Barrett actually considers all PIDs that satisfy Assumption $(*)$ of \citep{bertschinger2014quantifying}. \label{fn_assm_star}
		While this is not incorrect, it is \emph{incidental}, and considering Blackwellian PIDs is more precise.
		This difference is explained in Appendix~\ref{app_thm_gerdes}.}
	\item When $P_{MXY}$ is jointly Gaussian and $M$ is scalar, we always have that either $X$ is Blackwell sufficient for $M$ w.r.t.\ $Y$, or $Y$ is Blackwell sufficient w.r.t.\ $X$.
		However, this is not always the case when $M$ is a vector.
\end{enumerate}
Theorem~\ref{thm_barrett} follows by reading these steps in reverse order.
Thus, we observe that the concept of Blackwell sufficiency is the driving force in Barrett's result, although this is not explicitly recognized in that paper~\cite{barrett2015exploration}.
So we now proceed to formally define and understand Blackwell sufficiency.

\vspace{3pt}
\textit{Blackwell sufficiency. }
The idea of Blackwell sufficiency comes from the field of statistical decision theory~\cite{blackwell1953equivalent, torgersen1991comparison}.
Suppose we are allowed to make inferences about $M$ by observing either $X$ or $Y$, but not both.
Then, informally speaking, $X$ is said to be \emph{Blackwell sufficient} for $M$ with respect to $Y$ if, on average, we can make equal or better inferences about $M$ (w.r.t. \emph{any} loss function) by observing $X$ rather than $Y$.
This forms a useful operational interpretation for Blackwellian PIDs, which satisfy $UI_Y = 0$ if and only if $X$ is Blackwell sufficient for $M$ with respect to $Y$.
Thus, $UI_Y = 0$ if and only if $Y$ can provide no better information than $X$ for making inferences about $M$, which provides a concrete intuition for the meaning of \emph{unique} information~\cite{bertschinger2014quantifying}.
For simplicity, we define an equivalent~\cite{blackwell1953equivalent} notion of Blackwell sufficiency, which is more amenable to our setup here.
Given random variables $M$, $X$ and $Y$ with joint density $P_{MXY}$, Blackwell sufficiency is formally defined for the stochastic transformations $P_{X|M}$ and $P_{Y|M}$ (also called ``channels'' from $M$ to $X$ and from $M$ to $Y$ respectively):%
\footnote{Note that Blackwell sufficiency ignores potential dependencies between $X$ and $Y$ given $M$, i.e., it does not depend upon $P_{XY|M}$ in its entirety; it only depends on the $X$- and $Y$-marginals of $P_{XY|M}$.}
\begin{definition}[Blackwell sufficiency: $\suff_M$] \label{def_blackwell}
	We say that a channel $P_{X|M}$ is \emph{Blackwell sufficient} for another channel $P_{Y|M}$ (denoted $X \suff_M Y$) if $\exists\; P_{Y'|X} \in \mathcal C(\mathsf{Y} \given \mathsf{X})$ such that
	\begin{equation} \label{eq_blackwell_suff}
		P_{Y'|X}\circ P_{X|M} \;=\; P_{Y|M},
	\end{equation}
	where $\mathcal C(\sampy|\sampx)$ is the set of all channels from $\sampx$ to $\sampy$, and $\circ$ represents channel composition, i.e.\ $\forall\; m \in \sampm, y \in \sampy$,
	\begin{equation}
		(P_{Y'|X}\circ P_{X|M})(y|m) \coloneqq \int P_{Y'|X}(y|x) P_{X|M}(x|m)\,dx.
	\end{equation}
\end{definition}
Intuitively, $X \suff_M Y$ means that we can generate a new random variable $Y'$ from $X$ (using the stochastic transformation $P_{Y'|X}$), so that the effective channel from $M$ to $Y'$ is equivalent to the original channel from $M$ to $Y$ (i.e., for any $M=m$, samples of $Y'$ are statistically identical to samples of $Y$).
\begin{definition}[Blackwellian PID] \label{prop_blackwell}
	We say that a bivariate PID of $I\bigl(M ; (X, Y)\bigr)$ is \emph{Blackwellian} if $Y$ has zero unique information about $M$ with respect to $X$ if and only if $X \suff_M Y$.
\end{definition}

Recalling our main objective: we wish to find efficient ways of computing the fully multivariate Gaussian PID.
Based on step (C) in the list above, our first goal is to characterize the conditions under which $X \suff_M Y$ for fully multivariate Gaussian $P_{MXY}$.
Then, steps (B) and (A) would follow, ending with the closed form expression of the MMI-PID.

\vspace{3pt}
\textit{Characterizing Blackwell sufficiency for Gaussians. }
First, we parameterize the multivariate Gaussian distribution $P_{MXY}$.
Let the sample spaces of $M$, $X$ and $Y$ respectively be $\sampm = \mathbb R^{d_M}$, $\sampx = \mathbb R^{d_X}$ and $\sampy = \mathbb R^{d_Y}$.
Let the joint distribution be given by $P_{MXY} = \mathcal N(0, \Sigma)$.
Since Blackwell sufficiency ignores the dependence between $X$ and $Y$ conditional on $M$, we define:
\begin{alignat}{3}
	M &\sim P_M &&= \mathcal{N}(0,\sigm) \\
	X \given M &\sim P_{X|M} &&= \mathcal{N}(\hx M, \,\sigx) \\
	Y \given M &\sim P_{Y|M} &&= \mathcal{N}(\hy M, \,\sigy)
\end{alignat}
Here, $\hx$, $\hy$, $\sigx$ and $\sigy$ can all be extracted from the full covariance matrix, $\Sigma$.
Also, note that we \emph{do not} assume that $X \indept Y \given M$; this dependence is captured in $\Sigma$.
\begin{remark} \label{rem_full_rank}
	Without loss of generality, we assume that $\mathbb E[M] = 0$.
	We also assume that the noise covariance matrices $\sigx$ and $\sigy$ are full rank, which need not always be true (e.g., in a noiseless channel).
	However, this assumption is required to keep mutual information values finite, as explained in Appendix~\ref{app_misc_derivs}.
\end{remark}
We can now characterize the conditions for Blackwell sufficiency in fully multivariate Gaussian distributions.
\begin{theorem}[Blackwell Sufficiency for Multivariate Gaussians] \label{thm_gerdes}
	For jointly Gaussian random vectors $M$, $X$ and $Y$, $X \suff_M Y$ if and only if
	\begin{equation} \label{eq_gauss_bs_cond}
		\hx^\T \sigx^{-1} \hx \suff \hy^\T \sigy^{-1} \hy,
	\end{equation}
	where $A \suff B$ means that $A - B$ is positive semidefinite (for positive semidefinite matrices $A$ and $B$).
\end{theorem}
A proof of this theorem is provided in Appendix~\ref{app_thm_gerdes}\ifarxiv{}{ in the full version of this paper~\citep{full_version}}.
The result relies on a connection between concepts from two disparate fields: Blackwell sufficiency and the \emph{stochastic degradedness} of broadcast channels.%
\footnote{This connection was mentioned in passing by Raginsky~\cite{raginsky2011shannon}, but we prove the equivalence formally in Lemma~\ref{lem_bs_sd_equiv} (see Appendix~\ref{app_thm_gerdes}\ifarxiv{}{in the full version~\cite{full_version}}).
This connection has also never before been recognized in the context of PID.}
We then leverage prior work from the literature analyzing the stochastic degradedness of MIMO Gaussian broadcast channels~\cite{gerdes2015equivalence, shang2012noisy} to prove our result.
With this, we can formally state our extension of Barrett's result to fully multivariate Gaussians:
\begin{corollary}[Extension of Barrett's Result] \label{cor_gauss_pid}
	If $P_{MXY}$ is jointly Gaussian and satisfies \eqref{eq_gauss_bs_cond}, then any Blackwellian PID computed on $P_{MXY}$ reduces to the MMI-PID and can be written in closed form as shown in Remark~\ref{rem_gauss_mmi_pid}.
\end{corollary}
\begin{remark} \label{rem_whitening}
	Since $\sigx$ and $\sigy$ are assumed full rank (see Remark~\ref{rem_full_rank}), we can apply a whitening transform:
	\begin{equation} \label{eq_whitened}
		\widetilde H_X \coloneqq \sigx^{-\frac{1}{2}}\hx, \quad \widetilde H_Y \coloneqq \sigy^{-\frac{1}{2}}\hy.
	\end{equation}
	This simplifies \eqref{eq_gauss_bs_cond} to $\widetilde H_X^\T \widetilde H_X \suff \widetilde H_Y^\T \widetilde H_Y$.
	We assume henceforth that $\hx$ and $\hy$ have already been whitened, and hence $\sigx=\sigy=I$.
\end{remark}

\subsection{Efficiently Computing the Gaussian PID in General}
\label{sec_def_approx}

Next, we explore how the Gaussian PID can be computed (approximately) when $P_{MXY}$ does \emph{not} satisfy~\eqref{eq_gauss_bs_cond}.
For this, we consider another PID definition based on the statistical concept of \emph{deficiency}~\cite{banerjee2018unique}, which we call the $\delta$-PID.
\begin{definition}[$\delta$-PID \cite{banerjee2018unique}]\label{banerjee_pid}
	Let the (weighted output) \emph{deficiency}%
	\footnote{There are many ways to define deficiency; Raginsky~\cite{raginsky2011shannon} provides a number of these that consider the worst-case over $M$.
	We prefer an expectation over $M$, since $M$ is a random variable in our setup.}
	of $X$ with respect to $Y$ about $M$ be defined as%
	\footnote{The reason for this notation is that the deficiency of $X$ w.r.t.\ $Y$ translates to the unique information present in $Y$ and not in $X$.}
	\begin{equation}\label{eq_deficiency}
		\df{Y}{X} \coloneqq \quad \inf_{\mathclap{\vphantom{X^{X^X}} P_{Y'|X} \,\in\, \mathcal C(\sampy|\sampx)}} \quad \mathbb{E}_{P_{M}}\bigl[D(P_{Y|M} \,\Vert\, P_{Y'|X}\circ P_{X|M})\bigr],
	\end{equation}
	where $D(\cdot\Vert \cdot)$ is the KL-divergence. Then, define the deficiency-based redundant information about $M$ in $X$ and $Y$ as:
	\begin{equation} \label{eq_ri_delta}
		\begin{aligned}
			RI_{\delta}(M : X ; Y) \coloneqq \min\bigl\{&{} I(M ; X) - \df{X}{Y},\\
														&{} I(M ; Y) - \df{Y}{X}\bigr\}.
		\end{aligned}
	\end{equation}
\end{definition}
\noindent As with the MMI-PID, equations \eqref{eq_pid}--\eqref{eq_pid_y} fully determine the remaining components of the $\delta$-PID---$UI_{\delta X}$, $UI_{\delta Y}$ and $SI_\delta$.

Deficiency finds the channel $P_{Y'|X}$ that minimizes the expected divergence between $P_{Y|M}$ and $P_{Y'|X} \circ P_{X|M}$, where the expectation is over $M$.
The divergence goes to zero if and only if $P_{Y|M} = P_{Y'|X} \circ P_{X|M}$, i.e.\ if $X \suff_M Y$, thus the $\delta$-PID is also Blackwellian~\cite{torgersen1991comparison}.
Thus, deficiency directly measures how far $P_{MXY}$ is from being Blackwell \emph{sufficient} (as stated in~\eqref{eq_blackwell_suff}), and therefore forms a natural measure of unique information.
Equation~\eqref{eq_ri_delta} plays the role of symmetrizing the redundant information, since $I(M ; X) - \df{X}{Y}$ and $I(M ; Y) - \df{Y}{X}$ are not necessarily equal.

Computing the deficiency as given by \eqref{eq_deficiency} is challenging because it requires an optimization over all continuous conditional distributions $P_{Y'|X}$.
So, we consider the restricted problem where $P_{Y'|X}$ lies in the set of linear additive Gaussian noise channels $\mathcal{C}_G(\sampy\given\sampx)\subset \mathcal{C}(\sampy\given\sampx)$.
Then, $P_{Y'|X}$ can be parameterized in terms of its channel gain and noise covariance matrices, $T\in\mathbb{R}^{d_Y \times d_X}$ and $\sigt\in\mathbb{R}^{d_Y\times d_Y}$, $\sigt \suff 0$, as $P_{Y'|X} \sim \mathcal N(TX, \Sigma_T)$.
Thus, the \emph{Gaussian deficiency} can be defined as follows:
\begin{align}
	\MoveEqLeft[1] \dfg{Y}{X} \notag \\
	&\coloneqq \inf_{P_{Y'|X}\in \mathcal{C}_G(\sampy\mid\sampx)} \mathbb{E}_{P_{M}}\left[D(P_{Y\mid M} \,\Vert\, P_{Y'|X}\circ P_{X\mid M})  \right] \label{eq_gaussdef1} \\
	&= \inf_{T,\sigt\suff 0} \frac{1}{2}\biggl[\begin{aligned}[t]
			&\ \mathbb{E}_{P_M}\Bigl[\bigl\lVert(T\hx-\hy)M\bigr\rVert_{\sigt + TT^\T}^2\Bigr] \\
			&+ \operatorname{Tr}\left\{(\sigt + TT^\T)^{-1}\right\} \\
			&+ \log \mathrm{det}(\sigt + TT^\T) - d_Y
		\biggr],\end{aligned} \label{eq_gaussdef2}
\end{align}
where $\Vert a \Vert_B^2 \coloneqq a B^{-1} a^\T$ is the squared Mahalanobis distance, and \eqref{eq_gaussdef2} follows from the expression for the KL-divergence of normal distributions (derived in Appendix~\ref{app_misc_derivs}) and using Remark~\ref{rem_whitening}.
Unfortunately, we cannot prove that this problem is convex, since $\operatorname{Tr}\{(\sigt + TT^\T)^{-1}\}$ is the composition of a convex function ($\operatorname{Tr}\{(\cdot)^{-1}\}$) with a non-monotonic function of $T$.
Therefore, we propose a convex approximation of \eqref{eq_gaussdef2} to find an \emph{approximate} minimizer $\widehat{P}_{Y'|X}{\,=\,}\mathcal{N}(\widehat{T} X,\sigth)$:
\begin{definition}[$\widehat\delta_G$ and the $\widehat\delta_G$-PID] \label{def_approx_gaussdef}
	Let
	\begin{align}
		\widehat{T} &\coloneqq \underset{T}{\operatorname{argmin}} \;\; \mathbb{E}_{P_M} \left[ \bigl\lVert (T\hx - \hy)M \bigr\rVert_{I + \hy \sigm \hy^\T}^2 \right] \label{eq_th} \\
					&\hphantom{\coloneqq \vphantom{A}} \; \textup{s.t.} \;\; I + \hy\sigm\hy^\T - T (I+\hx\sigm\hx^\T) T^\T \suff 0 \notag \\[5pt]
					\sigth &\coloneqq I + \hy\sigm\hy^\T - \widehat{T}(I+\hx\sigm\hx^\T)\widehat{T}^\T \label{eq_sigth}
	\end{align}
	Then, we define the approximate deficiency to be:
	\begin{align} \label{eq_dfh}
		\dfh{Y}{X} \coloneqq \frac{1}{2}\bigg[
			&\mathbb{E}_{P_M}\Bigl[\bigl\lVert( \widehat{T}\hx-\hy)M\bigr\rVert_{\sigth + \widehat{T}\widehat{T}^\T}^2\Bigr] \notag \\
			&+ \operatorname{Tr}\bigl\{(\sigth + \widehat{T}\widehat{T}^\T)^{-1}\bigr\} \notag \\
			&+ \log\det(\sigth + \widehat{T}\widehat{T}^\T) - d_Y
		\smash{\bigg]}
	\end{align}
	Substituting $\widehat\delta_G$ into \eqref{eq_ri_delta}, we obtain an approximation of $RI_\delta$, and hence of the $\delta$-PID, which we call the $\widehat\delta_G$-PID.
\end{definition}
For brevity, we use $\widehat{UI}_X$, $\widehat{UI}_Y$, $\widehat{RI}$ and $\widehat{SI}$ to refer to the constituent atoms of the $\widehat\delta_G$-PID.

The derivation of the convex formulation in~\eqref{eq_th} has three main steps, which we summarize here and describe in detail in Appendix~\ref{app_justification}\ifarxiv{}{ in the full version~\citep{full_version}}.
\begin{enumerate}[nosep, leftmargin=*]
    \item First, we obtain a condition on $\sigt$ (in terms of $T$), which is locally optimal. This condition allows us to reduce the optimization problem from two variables ($\sigt$ and $T$) to just one ($T$~alone). This step is \emph{exact}, and is shown in Proposition~\ref{prop_gaussdef_simplifcn} in Appendix~\ref{app_justification}.
    \item Next, we reinterpret and approximate the objective to significantly simplify its functional form, while attempting to minimize the same entities.
    \item Lastly, as part of reducing the optimization problem to a single variable, the constraint $\sigt \suff 0$ is now replaced by a positive semidefiniteness constraint in terms of $T$. We approximate and simplify this constraint to make it more amenable for a convex program to handle, which yields the condition shown in~\eqref{eq_sigth}.
\end{enumerate}
This approximation can be shown to satisfy certain desirable properties.
Firstly, $\widehat\delta_G$ is well defined, i.e., $\sigth + \widehat{T}\widehat{T}^\T$ is guaranteed to be invertible when using the approximation of Definition~\ref{def_approx_gaussdef}.
This is formally stated and proved in Lemma~\ref{lemma_valid}, which can be found in Appendix~\ref{app_lemma_valid}\ifarxiv{}{ in the full version~\citep{full_version}}.

Secondly, the estimate in \eqref{eq_dfh} provably coincides with the true deficiency in the following circumstances:
\begin{proposition}\label{prop_extremes}
	For jointly Gaussian $P_{MXY}$,
	\begin{align}
		\widehat\delta_G(M : Y \sm X) = 0 \,&\Leftrightarrow\, \delta(M : Y \sm X) = 0 \label{eq_zero_def2}\\
		\delta(M{:\,}Y{\sm}X)=I(M;Y) \,&\Rightarrow\, \widehat\delta_G(M{:\,}Y{\sm}X)=I(M;Y) \label{eq_mi_def2}
	\end{align}
\end{proposition}
\noindent A proof is given in Appendix~\ref{app_prop_extremes}\ifarxiv{}{ in the full version~\cite{full_version}}.

Thirdly, since $\widehat{P}_{Y'|X}\in \mathcal{C}_G(\sampy\given\sampx) \subset \mathcal{C}(\sampy\given\sampx)$, we have:
\begin{equation}\label{eq_def_ineqs}
\dfh{Y}{X} \ge \dfg{Y}{X} \ge \df{Y}{X}
\end{equation}
These inequalities show that the $\widehat\delta_G$-PID bounds the true $\delta$-PID:
\begin{proposition}\label{prop_pid_bounds}
For jointly Gaussian random vectors $M$, $X$ and $Y$, the $\widehat{\delta}_G$-PID bounds the $\delta$-PID:
\begin{equation} \label{eq_bounds}
    \begin{aligned}
		\widehat{UI}_X & \ge {UI}_{\delta X} & \widehat{RI} & \le {RI}_\delta \\
		\widehat{UI}_Y & \ge {UI}_{\delta Y} & \widehat{SI} & \le {SI}_\delta
    \end{aligned}
\end{equation}
\end{proposition}

\section{Empirical Results}
\label{sec_simulations}

Having established and theoretically justified a framework for approximating the $\delta$\nobreakdash-PID, we address three questions using simulations:
(Q1) Does $\widehat\delta_G$ yield a \emph{non-negative} PID for jointly Gaussian distributions?
(Q2) Is the $\widehat\delta_G$-PID consistent with the results of Barrett~\cite{barrett2015exploration} for the case of univariate $M$?
(Q3) Does the $\widehat\delta_G$-PID meet our intuition for different $d_M$, $d_X$ and $d_Y$? %
(Q4) How good is the approximation provided by the $\widehat\delta_G$-PID?

To address these questions, we sampled 80,000 joint covariance matrices $\Sigma\in\mathbb{R}^{d\times d}$ ($d{\,=\,}d_M {\,+\,} d_X {\,+\,} d_Y$) from a standard Wishart distribution, which fully characterized $P_{MXY}$.
We sampled 20,000 matrices each from four sampling schemes:
\begin{enumerate}[label=(S\arabic*), nosep, leftmargin=*]
    \item $d_M \sim \operatorname{Unif}\{1 \dots 10\}$ and $d_X=d_Y=d_M$
    \item $d_M \sim \operatorname{Unif}\{1 \dots 9\}$ and $d_X,d_Y \overset{\text{iid}}{\sim} \operatorname{Unif}\{d_M+1 \ldots 10\}$
    \item $d_M \sim \operatorname{Unif}\{2 \dots 10\}$ and $d_X,d_Y \overset{\text{iid}}{\sim} \operatorname{Unif}\{1 \ldots d_M-1\}$
    \item $\vphantom{\overset{\text{iid}}{\sim}}d_M \sim \operatorname{Unif}\{2 \dots 9\}$, $d_X \sim \operatorname{Unif}\{1 \ldots d_M-1\}$, and $d_Y \sim \operatorname{Unif}\{d_M+1 \ldots 10\}$
\end{enumerate}
Without loss of generality, we took $d_X \le d_Y$ (i.e., we switched their values if $d_Y < d_X$).
Simulations were performed on a 2.4 GHz 8-Core Intel Core i9 processor and took 188.8 minutes to complete.
Details of the experimental setup and implementation are provided in Appendices \ref{app_cvx} and \ref{app_exp} \ifarxiv{}{\citep{full_version}} and the code is available on Github~\citep{code}.

First, we found that all estimated PID components were non-negative for \emph{every one} of the 80,000 sampled distributions.
This suggests an affirmative answer to (Q1).%
Secondly, we found that in every distribution with $d_M=1$ ($N=4277$), only one of $X$ or $Y$ had unique information.
Therefore, the $\widehat\delta_G$-PID is consistent with Barrett's result in our experiments, answering (Q2) in the affirmative.
Of the remaining 75,723 distributions for which $d_M > 1$, $X$ and $Y$ \emph{both} had unique information in 93.6\%
Thus, the overwhelming majority of Gaussian distributions with vector $M$ \emph{do not} satisfy Blackwell sufficiency and do not reduce to the MMI-PID, justifying the need for our approximation-based approach.

To address (Q3), we visualize the distribution of unique, redundant, and synergistic information across all four sampling schemes in Figure~\ref{fig_simplex}.
Since the scale of the PID components varies with $I(M;(X,Y))$, we consider the normalized PID quantities $\widebar{UI}_X$, $\widebar{UI}_Y$, $\widebar{RI}$, and $\widebar{SI}$ obtained by dividing each PID component by $I(M;(X,Y))$.
These normalized values are all non-negative and satisfy $\widebar{UI}_X+\widebar{UI}_Y+\widebar{RI}+\widebar{SI}=1$.
In Figure~\ref{fig_simplex} we represent each distribution by its location on a 3-simplex that characterizes the proportion of $I(M;(X,Y))$ that is accounted for by each PID component.
Below the 3-simplexes for each sampling scheme, we also show the distribution of $\widebar{UI}_X$, $\widebar{UI}_Y$, $\widebar{RI}$, and $\widebar{SI}$ in the form of box plots.

Figure \ref{fig_simplex} shows how our approximate PID for multivariate Gaussians meets our intuitive expectations.
Firstly, since $d_Y \geq d_X$, $Y$ tends to have more unique information than $X$ (see box plots), excepting (S1) where $d_X = d_Y$.
Secondly, (S2) with $d_M < d_X, d_Y$ closely mimics the scalar-$M$ case, rarely exhibiting unique information in both $X$ and $Y$ simultaneously.
This is seen in the simplex plot, wherein the isosceles triangle forming the lower third of the simplex (with vertices at $\widebar{UI}_X$, $\widebar{UI}_Y$ and the centroid of the simplex) is almost completely devoid of points.
Thirdly, (S3) and (S4) have large amounts of unique information in $Y$ (as seen in the box plots), since $d_M$ is large and provides more dimensions of $M$ that $Y$ can uniquely capture.

Synergy has a strong prevalence under all four sampling schemes (see box plots), with the greatest prevalence in (S2), when $d_M < d_X, d_Y$.
This agrees with the intuition provided by Barrett~\citep{barrett2015exploration}, who mentions the prevalence of synergy for Gaussians with $d_M = 1$.
Redundancy is never particularly prevalent---this may be due to how we sample random covariance matrices, which reduces the likelihood that $X$ and $Y$ capture the same dimensions of $M$.
(S3) and (S4) have the least redundancy (as observed in the box plots), since $X$ has fewer dimensions than $M$ in both these cases.

To address (Q4), unfortunately, there are no alternative estimators of the $\delta$-PID that can serve as potential baselines.
Nor are there alternative estimators of Blackwellian PIDs for Gaussian variables (excepting the scalar-$M$ case, for which Proposition~\ref{prop_extremes} already provides guarantees).
Therefore, we compare the $\widehat\delta_G$-PID with the best available estimator~\cite{banerjee2018computing} for the $\sim$-PID.
These results are explained in the caption of Figure~\ref{fig_poiss_vs_gauss}, with further details in Appendix~\ref{app_poiss_vs_gauss}\ifarxiv{}{in the full version of this paper~\cite{full_version}}.

\begin{figure}[t]
    \centering
	\includegraphics[width=\linewidth]{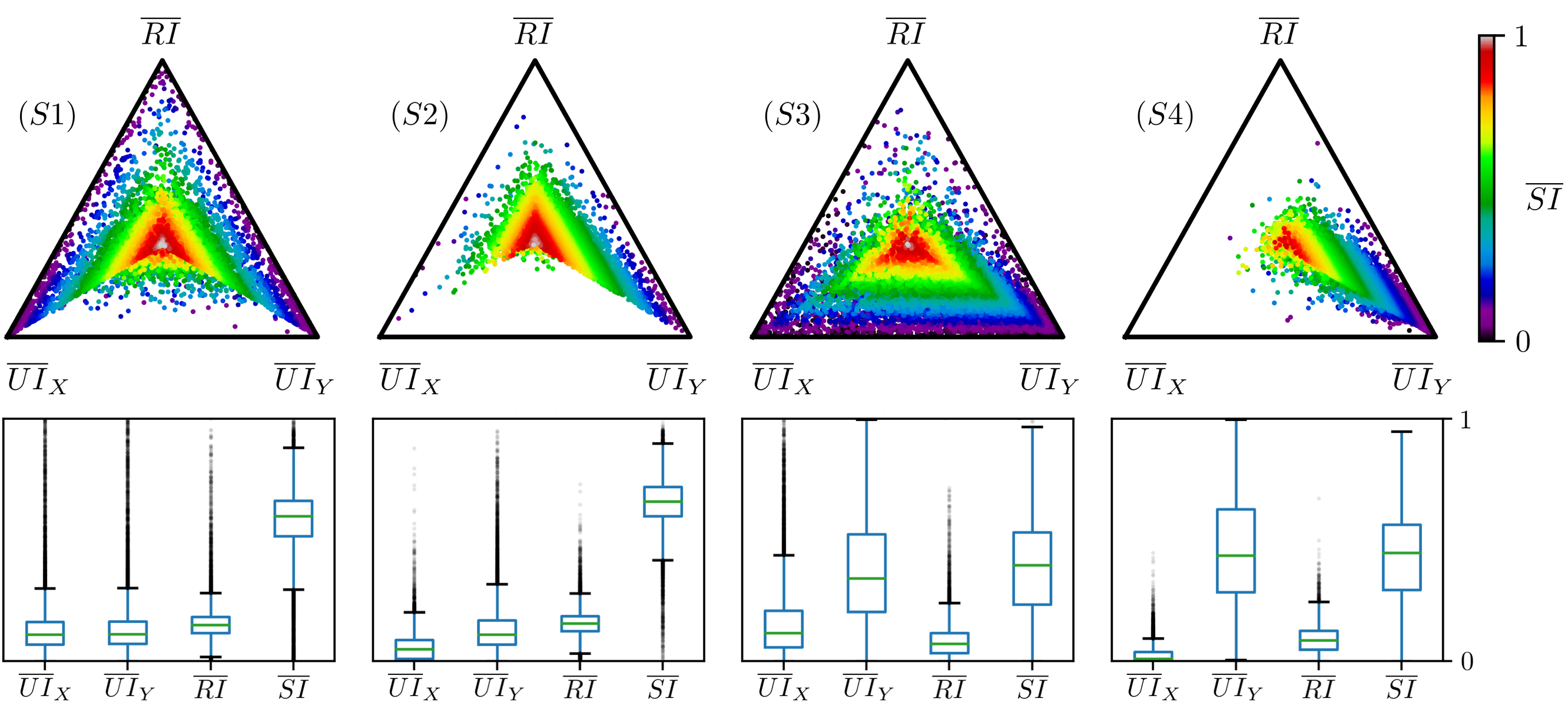}\\%
    \caption{Distribution of normalized unique, redundant, and synergistic information for Gaussian systems sampled from a standard Wishart distribution: (S1)~$d_M=d_X=d_Y$; (S2) $d_M < d_X\le d_Y$; (S3) $d_X\le d_Y< d_M$; and (S4)~$d_X < d_M < d_Y$.
		Top row: scatter plots of the computed $\protect\widebar{UI}_X$, $\protect\widebar{UI}_Y$, $\protect\widebar{RI}$ and $\protect\widebar{SI}$ on the 3-simplex (3D views in Appendix~\ref{app_3d_views}\protect\ifarxiv{}{ in the full version~\citep{full_version}}).
		Each point is a single sampled Gaussian system.
		Bottom row: box plots showing the relative prevalence of each PID atom.
		The box shows the median and the first and third quartiles, while the whiskers extend to 1.5 times the interquartile range (difference of third and first quartiles).
	}
    \label{fig_simplex}
\end{figure}

\begin{figure}[t]
	\centering
	\includegraphics[width=0.5\linewidth]{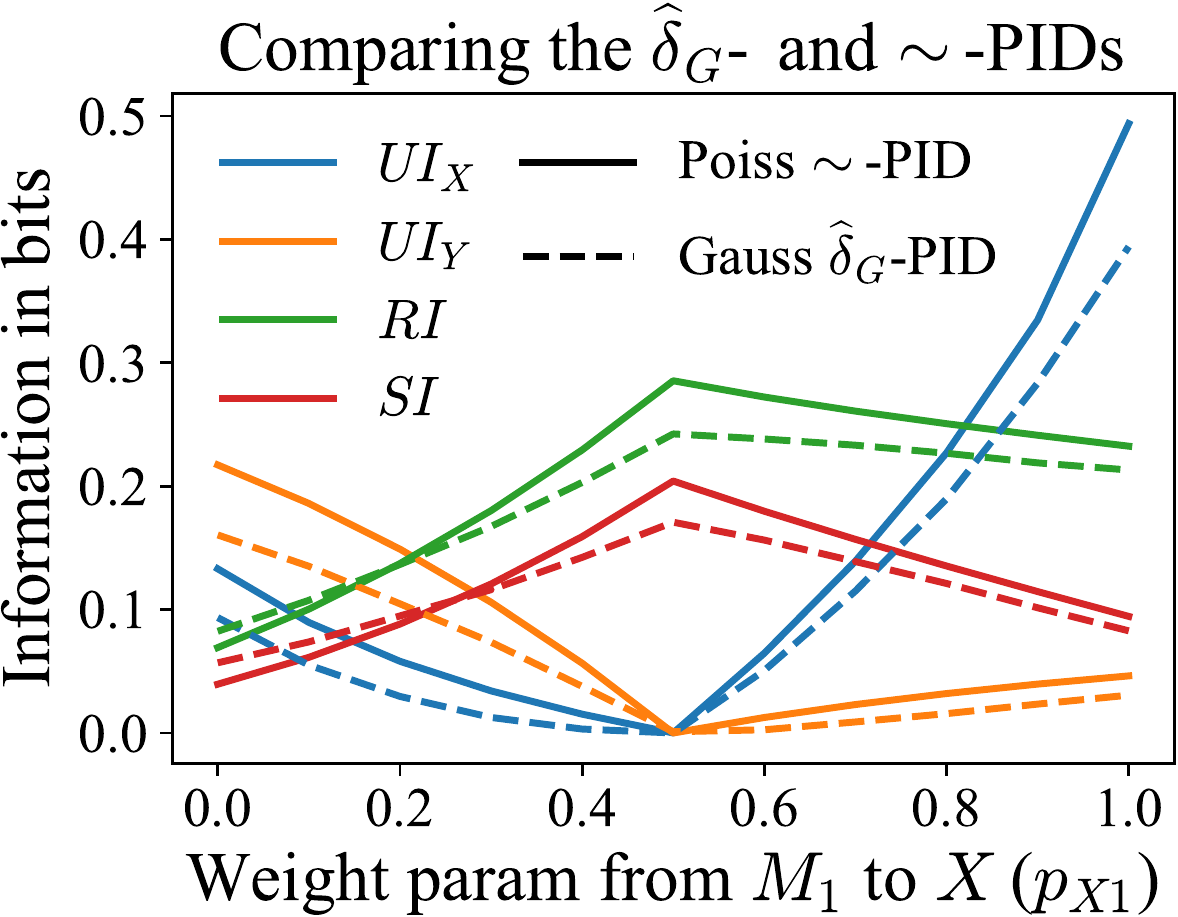}
	\caption{A comparison of the $\sim$-PID evaluated~\citep{banerjee2018computing} on a multivariate Poisson distribution against the $\widehat\delta_G$-PID evaluated on the covariance matrix of the same distribution.
		We take $M = [M_1, M_2]$, with $M_1, M_2 \sim$ iid.~Poiss$(2)$ and $X \sim$ Binom$(M_1, p_{X1}) +$ Binom$(M_2, p_{X2}) +$ Poiss$(1)$.
		$Y$ is defined similar to $X$.
		$p_{X1}$ is varied between 0.0 and 1.0 in increments of 0.1, while $p_{X2}$, $p_{Y1}$ and $p_{Y2}$ are kept fixed at 0.5.
		Despite differences in the PID definitions, errors in approximating the Poisson distribution as a Gaussian, and errors in the two estimates, there is a reasonable agreement between the two PIDs.
	}
	\label{fig_poiss_vs_gauss}
\end{figure}

\section{Summary and Open Questions}
\label{sec_discussion}

Barrett~\cite{barrett2015exploration} showed that a Blackwellian PID (the $\sim$-PID) can be easily computed for Gaussian $P_{MXY}$ with scalar $M$.
Theorem~\ref{thm_gerdes} shows that this extends to vector $M$ only under a specific condition, which likely holds only for a small minority of Gaussian distributions.
Therefore, we provided a convex optimization framework to approximate the $\delta$-PID, to enable applications of the PID on multivariate Gaussian data.

Our framework forms the first efficient method for computing Blackwellian PIDs (even approximately) for high-dimensional distributions.
Although we have not theoretically proved that the $\widehat\delta_G$-PID is non-negative, we found no negative instances empirically.
Several open questions remain: Can deficiency be computed exactly?
When does our approximation agree with the true deficiency (beyond Proposition~\ref{prop_extremes})?
Can we characterize Blackwell sufficiency for general probability distributions (e.g., see~\citep{makur2018comparison})?
Finally, applying the $\widehat\delta_G$-PID in practice also requires studying its statistical properties when estimating covariance matrices from data and providing statements of confidence.

\IEEEtriggeratref{23}
\bibliographystyle{IEEEtran}
\bibliography{IEEEabrv,references}

\begin{thebibliography}{10}
\providecommand{\url}[1]{#1}
\csname url@samestyle\endcsname
\providecommand{\newblock}{\relax}
\providecommand{\bibinfo}[2]{#2}
\providecommand{\BIBentrySTDinterwordspacing}{\spaceskip=0pt\relax}
\providecommand{\BIBentryALTinterwordstretchfactor}{4}
\providecommand{\BIBentryALTinterwordspacing}{\spaceskip=\fontdimen2\font plus
\BIBentryALTinterwordstretchfactor\fontdimen3\font minus
  \fontdimen4\font\relax}
\providecommand{\BIBforeignlanguage}[2]{{%
\expandafter\ifx\csname l@#1\endcsname\relax
\typeout{** WARNING: IEEEtran.bst: No hyphenation pattern has been}%
\typeout{** loaded for the language `#1'. Using the pattern for}%
\typeout{** the default language instead.}%
\else
\language=\csname l@#1\endcsname
\fi
#2}}
\providecommand{\BIBdecl}{\relax}
\BIBdecl

\bibitem{barrett2015exploration}
A.~B. Barrett, ``Exploration of synergistic and redundant information sharing
  in static and dynamical {G}aussian systems,'' \emph{Physical Review E},
  vol.~91, no.~5, p. 052802, 2015.

\bibitem{williams2010nonnegative}
P.~L. Williams and R.~D. Beer, ``Nonnegative decomposition of multivariate
  information,'' \emph{arXiv preprint arXiv:1004.2515}, 2010.

\bibitem{bertschinger2014quantifying}
N.~Bertschinger, J.~Rauh, E.~Olbrich, J.~Jost, and N.~Ay, ``Quantifying unique
  information,'' \emph{Entropy}, vol.~16, no.~4, pp. 2161--2183, 2014.

\bibitem{harder2013bivariate}
M.~Harder, C.~Salge, and D.~Polani, ``Bivariate measure of redundant
  information,'' \emph{Physical Review E}, vol.~87, no.~1, p. 012130, 2013.

\bibitem{banerjee2018unique}
P.~K. Banerjee, E.~Olbrich, J.~Jost, and J.~Rauh, ``Unique informations and
  deficiencies,'' in \emph{2018 56th Annual Allerton Conference on
  Communication, Control, and Computing (Allerton)}.\hskip 1em plus 0.5em minus
  0.4em\relax IEEE, 2018, pp. 32--38.

\bibitem{niu2019measure}
X.~Niu and C.~J. Quinn, ``A measure of synergy, redundancy, and unique
  information using information geometry,'' in \emph{2019 IEEE International
  Symposium on Information Theory (ISIT)}.\hskip 1em plus 0.5em minus
  0.4em\relax IEEE, 2019, pp. 3127--3131.

\bibitem{finn2018pointwise}
C.~Finn and J.~T. Lizier, ``Pointwise partial information decomposition using
  the specificity and ambiguity lattices,'' \emph{Entropy}, vol.~20, no.~4, p.
  297, 2018.

\bibitem{kay2018exact}
J.~W. Kay and R.~A. Ince, ``Exact partial information decompositions for
  gaussian systems based on dependency constraints,'' \emph{Entropy}, vol.~20,
  no.~4, p. 240, 2018.

\bibitem{lizier2018information}
J.~T. Lizier, N.~Bertschinger, J.~Jost, and M.~Wibral, ``Information
  decomposition of target effects from multi-source interactions: perspectives
  on previous, current and future work,'' \emph{Entropy}, vol.~20, no.~4, p.
  307, 2018.

\bibitem{full_version}
\BIBentryALTinterwordspacing
Full version of this paper with appendices. [Online]. Available:
  \url{https://praveenv253.github.io/assets/doc/papers/2022--isit--full-paper.pdf}
\BIBentrySTDinterwordspacing

\bibitem{colenbier2020disambiguating}
N.~Colenbier, F.~Van~de Steen, L.~Q. Uddin, R.~A. Poldrack, V.~D. Calhoun, and
  D.~Marinazzo, ``Disambiguating the role of blood flow and global signal with
  partial information decomposition,'' \emph{Neuroimage}, vol. 213, p. 116699,
  2020.

\bibitem{boonstra2019information}
T.~W. Boonstra, L.~Faes, J.~N. Kerkman, and D.~Marinazzo, ``Information
  decomposition of multichannel emg to map functional interactions in the
  distributed motor system,'' \emph{NeuroImage}, vol. 202, p. 116093, 2019.

\bibitem{krohova2019multiscale}
J.~Krohova, L.~Faes, B.~Czippelova, Z.~Turianikova, N.~Mazgutova, R.~Pernice,
  A.~Busacca, D.~Marinazzo, S.~Stramaglia, and M.~Javorka, ``Multiscale
  information decomposition dissects control mechanisms of heart rate
  variability at rest and during physiological stress,'' \emph{Entropy},
  vol.~21, no.~5, p. 526, 2019.

\bibitem{timme2018tutorial}
N.~M. Timme and C.~Lapish, ``A tutorial for information theory in
  neuroscience,'' \emph{eneuro}, vol.~5, no.~3, 2018.

\bibitem{pica2017quantifying}
G.~Pica, E.~Piasini, H.~Safaai, C.~Runyan, C.~Harvey, M.~Diamond, C.~Kayser,
  T.~Fellin, and S.~Panzeri, ``Quantifying how much sensory information in a
  neural code is relevant for behavior,'' in \emph{Advances in Neural
  Information Processing Systems}, vol.~30, 2017.

\bibitem{gat1999synergy}
I.~Gat and N.~Tishby, ``Synergy and redundancy among brain cells of behaving
  monkeys,'' \emph{Advances in neural information processing systems}, pp.
  111--117, 1999.

\bibitem{schneidman2003synergy}
E.~Schneidman, W.~Bialek, and M.~J. Berry, ``Synergy, redundancy, and
  independence in population codes,'' \emph{Journal of Neuroscience}, vol.~23,
  no.~37, pp. 11\,539--11\,553, 2003.

\bibitem{brenner2000synergy}
N.~Brenner, S.~P. Strong, R.~Koberle, W.~Bialek, and R.~R. d. R.~v. Steveninck,
  ``Synergy in a neural code,'' \emph{Neural computation}, vol.~12, no.~7, pp.
  1531--1552, 2000.

\bibitem{venkatesh2020information}
P.~Venkatesh, S.~Dutta, and P.~Grover, ``Information flow in computational
  systems,'' \emph{IEEE Transactions on Information Theory}, vol.~66, no.~9,
  pp. 5456--5491, 2020.

\bibitem{scagliarini2020synergistic}
T.~Scagliarini, L.~Faes, D.~Marinazzo, S.~Stramaglia, and R.~N. Mantegna,
  ``Synergistic information transfer in the global system of financial
  markets,'' \emph{Entropy}, vol.~22, no.~9, p. 1000, 2020.

\bibitem{dutta2020information}
S.~Dutta, P.~Venkatesh, P.~Mardziel, A.~Datta, and P.~Grover, ``An
  information-theoretic quantification of discrimination with exempt
  features,'' in \emph{Proceedings of the AAAI Conference on Artificial
  Intelligence}, vol.~34, no.~04, 2020, pp. 3825--3833.

\bibitem{blackwell1953equivalent}
D.~Blackwell, ``Equivalent comparisons of experiments,'' \emph{The Annals of
  Mathematical Statistics}, pp. 265--272, 1953.

\bibitem{banerjee2018computing}
P.~K. Banerjee, J.~Rauh, and G.~Mont{\'u}far, ``Computing the unique
  information,'' in \emph{2018 IEEE International Symposium on Information
  Theory (ISIT)}.\hskip 1em plus 0.5em minus 0.4em\relax IEEE, 2018, pp.
  141--145.

\bibitem{torgersen1991comparison}
E.~Torgersen, \emph{Comparison of statistical experiments}.\hskip 1em plus
  0.5em minus 0.4em\relax Cambridge University Press, 1991, vol.~36.

\bibitem{raginsky2011shannon}
M.~Raginsky, ``Shannon meets {B}lackwell and {L}e {C}am: {C}hannels, codes, and
  statistical experiments,'' in \emph{2011 IEEE International Symposium on
  Information Theory Proceedings}.\hskip 1em plus 0.5em minus 0.4em\relax IEEE,
  2011, pp. 1220--1224.

\bibitem{gerdes2015equivalence}
L.~Gerdes, M.~Riemensberger, and W.~Utschick, ``On the equivalence of degraded
  {G}aussian {MIMO} broadcast channels,'' in \emph{WSA 2015; 19th International
  ITG Workshop on Smart Antennas}.\hskip 1em plus 0.5em minus 0.4em\relax VDE,
  2015, pp. 1--5.

\bibitem{shang2012noisy}
X.~Shang and H.~V. Poor, ``Noisy-interference sum-rate capacity for vector
  {G}aussian interference channels,'' \emph{IEEE transactions on information
  theory}, vol.~59, no.~1, pp. 132--153, 2012.

\bibitem{code}
\BIBentryALTinterwordspacing
Multivariate Gaussian PID package. [Online]. Available:
  \url{https://github.com/gabeschamberg/mvar-gauss-pid}
\BIBentrySTDinterwordspacing

\bibitem{makur2018comparison}
A.~Makur and Y.~Polyanskiy, ``Comparison of channels: {C}riteria for domination
  by a symmetric channel,'' \emph{IEEE Transactions on Information Theory},
  vol.~64, no.~8, pp. 5704--5725, 2018.

\bibitem{cover2012elements}
T.~M. Cover and J.~A. Thomas, \emph{Elements of Information Theory}.\hskip 1em
  plus 0.5em minus 0.4em\relax John Wiley \& Sons, 2012.

\bibitem{petersen2012matrix}
\BIBentryALTinterwordspacing
K.~B. Petersen and M.~S. Pedersen, \emph{The {M}atrix {C}ookbook}.\hskip 1em
  plus 0.5em minus 0.4em\relax Technical University of Denmark, 2012, version:
  November 15, 2012. [Online]. Available:
  \url{http://www2.compute.dtu.dk/pubdb/pubs/3274-full.html}
\BIBentrySTDinterwordspacing

\bibitem{kraskov2004estimating}
A.~Kraskov, H.~St{\"o}gbauer, and P.~Grassberger, ``Estimating mutual
  information,'' \emph{Physical review E}, vol.~69, no.~6, p. 066138, 2004.

\bibitem{laue2018computing}
S.~Laue, M.~Mitterreiter, and J.~Giesen, ``Computing higher order derivatives
  of matrix and tensor expressions,'' \emph{Advances in neural information
  processing systems}, vol.~31, 2018.

\bibitem{laue2020simple}
------, ``A simple and efficient tensor calculus,'' in \emph{Proceedings of the
  AAAI Conference on Artificial Intelligence}, vol.~34, no.~04, 2020, pp.
  4527--4534.

\bibitem{diamond2016cvxpy}
S.~Diamond and S.~Boyd, ``{CVXPY}: {A} {P}ython-embedded modeling language for
  convex optimization,'' \emph{Journal of Machine Learning Research}, vol.~17,
  no.~83, pp. 1--5, 2016.

\bibitem{agrawal2018rewriting}
A.~Agrawal, R.~Verschueren, S.~Diamond, and S.~Boyd, ``A rewriting system for
  convex optimization problems,'' \emph{Journal of Control and Decision},
  vol.~5, no.~1, pp. 42--60, 2018.

\bibitem{zhang2006schur}
F.~Zhang, \emph{The {S}chur complement and its applications}.\hskip 1em plus
  0.5em minus 0.4em\relax Springer Science \& Business Media, 2006, vol.~4.

\bibitem{ocpb:16}
\BIBentryALTinterwordspacing
B.~O'Donoghue, E.~Chu, N.~Parikh, and S.~Boyd, ``Conic optimization via
  operator splitting and homogeneous self-dual embedding,'' \emph{Journal of
  Optimization Theory and Applications}, vol. 169, no.~3, pp. 1042--1068, June
  2016. [Online]. Available: \url{http://stanford.edu/~boyd/papers/scs.html}
\BIBentrySTDinterwordspacing

\bibitem{scs}
------, ``{SCS}: Splitting conic solver, version 2.1.3,''
  \url{https://github.com/cvxgrp/scs}, Nov. 2019.

\end{thebibliography}

\section*{Acknowledgments}

The authors thank Christof Koch, Sanghamitra Dutta and Pulkit Grover for useful discussions.
We also thank the anonymous reviewers of previous drafts of this paper, whose comments helped us correct certain technical issues and substantially improved the presentation.

P.~Venkatesh was supported by the Shanahan Family Foundation Fellowship at the Interface of Data and Neuroscience at the Allen Institute and the University of Washington, supported in part by the Allen Institute.
G.~Schamberg was supported by the Picower Institute for Learning and Memory.
We wish to thank the Allen Institute founder, Paul G. Allen, for his vision, encouragement, and support.

\clearpage
\onecolumn
\begin{center}
	{\large\bfseries Partial Information Decomposition via Deficiency for Multivariate Gaussians} \\[4pt]
	{\normalsize Praveen Venkatesh and Gabriel Schamberg}
\end{center}

\appendices

\setlength{\abovedisplayskip}{\origabovedisplayskip}
\setlength{\belowdisplayskip}{\origbelowdisplayskip}
\setlength{\abovedisplayshortskip}{\origabovedisplayshortskip}
\setlength{\belowdisplayshortskip}{\origbelowdisplayshortskip}

\section{Miscellaneous derivations}
\label{app_misc_derivs}

\subsection{Derivation of the Closed-form Expression for the Gaussian MMI-PID from Remark~\ref{rem_gauss_mmi_pid}}

Remark~\ref{rem_gauss_mmi_pid} establishes a closed-form expression for the MMI-PID of Gaussian random variables.
First, we derive the expression for mutual information between two jointly Gaussian random vectors, as shown in equation~\eqref{eq_gauss_mmi_pid}.
\begin{align}
	I(M ; Z) &\overset{(a)}{=} h(M) - h(M \given Z) \\
	h(M) &\overset{(b)}{=} \frac{1}{2} \log \bigl( (2 \pi e)^{d_M} \mathrm{det}(\Sigma_M) \bigr) \\
	h(M \given Z) &\overset{(c)}{=} \frac{1}{2} \log \bigl( (2 \pi e)^{d_M} \mathrm{det}(\Sigma_{M|Z}) \bigr)
\end{align}
where (a) is the basic formula for mutual information~\cite[Ch.~2]{cover2012elements} and (b) is the formula for the differential entropy of a Gaussian random vector~\cite[Thm.~8.4.1]{cover2012elements}.
Here, $h(\cdot)$ represents differential entropy, while $\Sigma_M$ is the covariance matrix of $M$.
In (c), we have used the fact that $P_{M|Z}$ is also a Gaussian distribution for jointly Gaussian $P_{MZ}$, so the differential entropy follows the same formula as in (b), with the difference that we use the \emph{conditional} covariance matrix $\Sigma_{M|Z}$ corresponding to the conditional distribution $P_{M|Z}$.

Next, we observe that $\Sigma_{M|Z}$ can be written as:
\begin{equation}
	\Sigma_{M|Z} = \Sigma_M - \Sigma^{}_{MZ} \Sigma_Z^{-1} \Sigma_{MZ}^\T
\end{equation}
where $\Sigma_{MZ}$ is the cross-covariance matrix between $M$ and $Z$ (e.g., see \cite[Sec.~8.1.3]{petersen2012matrix}).

Thus, the expression for mutual information reduces to
\begin{align}
	I(M ; Z) &= \frac{1}{2} \log \bigl( (2 \pi e)^{d_M} \mathrm{det}(\Sigma_M) \bigr) - \frac{1}{2} \log \bigl( (2 \pi e)^{d_M} \mathrm{det}(\Sigma_{M|Z}) \bigr) \\
			 &= \frac{1}{2} \log \biggl( \frac{\mathrm{det}(\Sigma_M)}{\mathrm{det}(\Sigma_{M|Z})} \biggr) \\
			 &= \frac{1}{2} \log \biggl( \frac{\mathrm{det}(\Sigma_M)}{\mathrm{det}(\Sigma_M - \Sigma^{}_{MZ} \Sigma_Z^{-1} \Sigma_{MZ}^\T)} \biggr). \label{eq_gauss_mmi_pid2}
\end{align}
A different closed-form expression for the mutual information between two Gaussian random variables appears later, in Section~\ref{app_sd_of_broadcast_channels} of the Appendix. That expression has the more familiar form of $1/2 \log(1 + SNR)$, but requires a more explicit channel parameterization, which we avoid here.

We can now write out the MMI-PID for jointly Gaussian $P_{MXY}$ using the above expression.
First, we individually substitute $X$ and $Y$ for $Z$ in the RHS below, to evaluate the MMI-redundancy:
\begin{equation}
	RI_{\textup{MMI}}(M : X ; Y) = \min\{I(M ; X),I(M ; Y)\}.
\end{equation}
Then, we can compute the $MMI$-unique informations by subtracting out the redundancy from the respective mutual information:
\begin{align}
	UI_{\textup{MMI}}(M : X \setminus Y) &= I(M ; X) - RI_{\textup{MMI}}(M : X ; Y), \\
	UI_{\textup{MMI}}(M : Y \setminus X) &= I(M ; Y) - RI_{\textup{MMI}}(M : X ; Y).
\end{align}
And finally, the synergy can be computed by subtracting each of the aforementioned terms from the total mutual information:
\begin{equation}
	SI_{\textup{MMI}}(M : X ; Y) = I(M ; X, Y) - UI_{\textup{MMI}}(M : X \setminus Y) - UI_{\textup{MMI}}(M : Y \setminus X) - RI_{\textup{MMI}}(M : X ; Y).
\end{equation}
Here, the total mutual information $I(M ; X, Y)$ is computed by substituting the concatenated vector, $[X^\T, Y^\T]^\T$, in place of $Z$ in \eqref{eq_gauss_mmi_pid2}.

\subsection{Details of the Counterexample to Barrett's Result for Vector $M$}
\label{app_counterexample}

Recall the objective of the counterexample: we wish to show that when $M$ is a vector, the $\sim$-PID of Definition~\ref{def_tilde_pid} does not always reduce to the MMI-PID of Definition~\ref{def_mmi_pid}.
Consider again the setup of the counterexample: $M = [M_1, M_2]$, $X = M_1 + Z_1$ and $Y = M_2 + Z_2$, with $M_1$, $M_2$, $Z_1$, $Z_2 \sim \text{ i.i.d.\ } \mathcal N(0, 1)$.

First, we show what the MMI-PID looks like in this case:
\begin{align}
	I(M ; X) &= I(M_1, M_2 ; X) = I(M_1 ; X) + I(M_2 ; X \given M_1) \\
			 &= I(M_1 ; X) + 0,
\end{align}
since $(X, M_1) \indept M_2$.
Similarly, $I(M ; Y) = I(M_2 ; Y)$.
Now, by symmetry, we have that
\begin{equation}
	I(M_1 ; X) = I(M_1 ; M_1 + Z_1) \overset{(a)}{=} I(M_2 ; M_2 + Z_2) = I(M_2 ; Y),
\end{equation}
where (a) follows from the fact that $M_1$, $M_2$, $Z_1$ and $Z_2$ are all independent and identically distributed. Thus,
\begin{equation}
	RI_{\textup{MMI}}(M : X ; Y) = \min\{I(M ; X),I(M ; Y)\} = I(M_1 ; X) = I(M_2 ; Y).
\end{equation}
Consequently, \eqref{eq_pid_x} and \eqref{eq_pid_y} imply that $UI_{\textup{MMI}}(M : X \setminus Y) = 0$ and $UI_{\textup{MMI}}(M : Y \setminus X) = 0$.

On the other hand, the $\sim$-PID satisfies:
\begin{equation} \label{eq_indep_conds}
	X \indept Y \quad \text{and} \quad X \indept Y \given M.
\end{equation}
The latter follows because
\begin{align}
	I(X ; Y \given M) &= I(M_1 + Z_1 ; M_2 + Z_2 \given M_1, M_2) \\
	&= h(M_1 + Z_1 \given M_1, M_2) - h(M_1 + Z_1 \given M_1, M_2, M_2 + Z_2) \\
	&= h(Z_1) - h(Z_1) = 0
\end{align}
Now, \citep[Thm.~20]{bertschinger2014quantifying} and \citep[Lem.~21]{bertschinger2014quantifying} state that when the two conditions in \eqref{eq_indep_conds} hold, then $\widetilde{RI}(M : X ; Y) = 0$.
Therefore, it follows that $\widetilde{UI}_X$ and $\widetilde{UI}_Y$ are equal to $I(M ; X)$ and $I(M ; Y)$ respectively.

Thus, we see that the $\sim$-PID does not always reduce to the MMI-PID for vector $M$.

\subsection{Explanation of Remark~\ref{rem_full_rank}}

The mutual information between two random variables is invariant under invertible (and measurable) transforms applied to each of the individual variables (see, e.g., \citep[Appendix]{kraskov2004estimating}).
Changing the means of $M$, $X$ and $Y$ constitutes an invertible transformation and would not affect their mutual information.
Thus, we can assume without loss of generality that $\mathbb E[M] = 0$.

Further, we explicitly make the assumption that $\sigx$ and $\sigy$ are full rank in order to ensure that mutual informations remain finite.
For example, if $M$ and $X$ are scalar, and $X$ is the output of a unity-gain noiseless channel with $M$ as input, then we would have $X = M$ and $\sigx = \sigma_{X|M}^2 = 0$.
Thus,
\begin{equation}
	I(M ; X) = h(M) - h(M \given X) = h(M) - h(M \given M) = \infty,
\end{equation}
because the conditional differential entropy of $M$ given itself is negative infinity.
Intuitively, this happens because $M$ is a continuous random variable which contains an infinite number of (non-repeating) unknown bits in its decimal expansion.
``Giving'' $M$ tells us all of these bits, and infinitely reduces our uncertainty about $M$.

Extending this concept to vector channels, to keep mutual information finite, we require that $\sigx$ and $\sigy$ must \emph{not} be noiseless in any direction (i.e., $p_{X|M}$ and $p_{Y|M}$ cannot be degenerate), and hence must be full rank.
Technically, it suffices that $\sigx$ and $\sigy$ have no rank deficiency along \emph{those} dimensions of $X$ and $Y$ that \emph{interact} with $M$, i.e.,
\begin{align}
	v \in \mathbb R^{d_X} \text{ s.t. } \lVert v \rVert = 1, \; v^\T \sigx \, v = 0 \quad \Rightarrow \quad v^\T X \indept M.
\end{align}
However, for simplicity, we assume that $\sigx$ and $\sigy$ are entirely full rank.

\subsection{Derivation of the Expression for Gaussian Deficiency in Equation~\eqref{eq_gaussdef2}}

First, we derive the expression for the KL-divergence between two $n$-dimensional multivariate normal distributions.
For simplicity, suppose the distributions are given by $p \sim \mathcal N(\mu_1, \Sigma_1)$ and $q \sim \mathcal N(\mu_2, \Sigma_2)$.
Then,
\begin{align}
	\MoveEqLeft[2] D\bigl(p(x) \Vert q(x)\bigr) = \int p(x) \log \frac{p(x)}{q(x)} dx = \mathbb E_p \biggl[ \log \frac{p(X)}{q(X)} \biggr] \\[3pt]
	&\overset{(a)}{=} \mathbb E_p \biggl[ \log\biggl( \frac{1}{(2\pi)^{n/2} \det(\Sigma_1)^{1/2}} \exp\Bigl(-\frac{1}{2} \lVert X - \mu_1 \rVert^2_{\Sigma_1}\Bigr) \biggr)
	- \log\biggl( \frac{1}{(2\pi)^{n/2} \det(\Sigma_2)^{1/2}} \exp\Bigl(-\frac{1}{2} \lVert X - \mu_2 \rVert^2_{\Sigma_2}\Bigr) \biggr) \biggr] \\[3pt]
	&\overset{(b)}{=} \mathbb E_p \biggl[ \log \frac{\det(\Sigma_2)^{1/2}}{\det(\Sigma_1)^{1/2}} - \frac{1}{2} \lVert X - \mu_1 \rVert^2_{\Sigma_1} + \frac{1}{2} \lVert X - \mu_2 \rVert^2_{\Sigma_2} \biggr] \\[3pt]
	&\overset{(c)}{=} \frac{1}{2} \mathbb E_p \biggl[ \log \frac{\det(\Sigma_2)}{\det(\Sigma_1)} - (X - \mu_1)^\T \Sigma_1^{-1} (X - \mu_1) +  (X - \mu_2)^\T \Sigma_2^{-1} (X - \mu_2) \biggr] \\[3pt]
	&\overset{(d)}{=} \frac{1}{2} \mathbb E_p \biggl[ \log \frac{\det(\Sigma_2)}{\det(\Sigma_1)} - \Tr\{(X - \mu_1)^\T \Sigma_1^{-1} (X - \mu_1)\} + \Tr\{(X - \mu_2)^\T \Sigma_2^{-1} (X - \mu_2)\} \biggr] \\[3pt]
	&\overset{(e)}{=} \frac{1}{2} \mathbb E_p \biggl[ \log \frac{\det(\Sigma_2)}{\det(\Sigma_1)} - \Tr\{\Sigma_1^{-1} (X - \mu_1) (X - \mu_1)^\T\} + \Tr\{\Sigma_2^{-1} (X - \mu_2) (X - \mu_2)^\T\} \biggr] \\[3pt]
	&\overset{(f)}{=} \frac{1}{2} \biggl[ \log \frac{\det(\Sigma_2)}{\det(\Sigma_1)} - \Tr\bigl\{\Sigma_1^{-1} \mathbb E_p [(X - \mu_1) (X - \mu_1)^\T] \bigr\} + \Tr\bigl\{\Sigma_2^{-1} \mathbb E_p [(X - \mu_2) (X - \mu_2)^\T] \bigr\} \biggr] \\[3pt]
	&\overset{(g)}{=} \frac{1}{2} \biggl[ \log \frac{\det(\Sigma_2)}{\det(\Sigma_1)} - \Tr\bigl\{\Sigma_1^{-1} \Sigma_1^{} \bigr\} + \Tr\bigl\{\Sigma_2^{-1} \bigl( \operatorname{Var}_p[X - \mu_2] + \mathbb E_p[X - \mu_2] \mathbb E_p[X - \mu_2]^\T \bigr) \bigr\} \biggr] \\[3pt]
	&\overset{(h)}{=} \frac{1}{2} \biggl[ \log \frac{\det(\Sigma_2)}{\det(\Sigma_1)} - \Tr\{I\} + \Tr\bigl\{\Sigma_2^{-1} \bigl(\Sigma_1 + (\mu_1 - \mu_2) (\mu_1 - \mu_2)^\T \bigr) \bigr\} \biggr] \\[3pt]
	&\overset{(i)}{=} \frac{1}{2} \biggl[ \log \frac{\det(\Sigma_2)}{\det(\Sigma_1)} - n + \Tr\bigl\{\Sigma_2^{-1} \Sigma_1 + \Sigma_2^{-1} (\mu_1 - \mu_2) (\mu_1 - \mu_2)^\T \bigr\} \biggr] \\[3pt]
	&\overset{(j)}{=} \frac{1}{2} \biggl[ \log \frac{\det(\Sigma_2)}{\det(\Sigma_1)} - n + \Tr\bigl\{\Sigma_2^{-1} \Sigma_1 \bigr\} + (\mu_1 - \mu_2)^\T \Sigma_2^{-1} (\mu_1 - \mu_2) \bigr\} \biggr] \\[3pt]
	&\overset{(k)}{=} \frac{1}{2} \biggl[ \log \frac{\det(\Sigma_2)}{\det(\Sigma_1)} - n + \Tr\bigl\{\Sigma_2^{-1} \Sigma_1 \bigr\} + \lVert \mu_1 - \mu_2 \rVert_{\Sigma_2} \bigr\} \biggr] \label{eq_gauss_kldiv}
\end{align}
where in the above steps, we have:
\begin{enumerate}[label=(\alph*)]
	\item Expanded out the Gaussian distribution
	\item Expanded out the logs and canceled terms
	\item Expanded the Mahalanobis distance, $\Vert a \Vert_B^2 \coloneqq a B^{-1} a^\T$
	\item Used the fact that a scalar is equal to the trace of itself
	\item Used the fact that the trace is invariant under cyclic permutations
	\item Used linearity of trace to take the expectation inside the operator
	\item Noted that $\mathbb E_p \bigl[(X - \mu_1) (X - \mu_1)^\T\bigr]$ is simply $\Sigma_1$, and used the formula for variance in the last term.
	\item Resolved the variance and expectation, taken with respect to $p$
	\item Used the fact that the trace of the identity matrix is equal to the dimension of the matrix, $n$
	\item Used the cyclic-permutation invariance of trace again
	\item Used the formula for Mahalanobis distance again
\end{enumerate}
The expression for Gaussian deficiency follows by observing that $p = \mathcal N(\hy M, \, \sigy)$ and $q = \mathcal N(T \hx M, \, T \sigx T^\T + \sigt)$.
Further, we make use of Remark~\ref{rem_whitening} to set $\sigx = I$ and $\sigy = I$, so that
\begin{align}
	\mu_1 &= \hy M & \Sigma_1 &= I \\
	\mu_2 &= T \hx M & \Sigma_2 &= \sigt + TT^\T
\end{align}
Finally, substituting into \eqref{eq_gauss_kldiv}, we get
\begin{align}
	\MoveEqLeft[2] \mathbb{E}_{P_{M}}\left[D(P_{Y\mid M} \,\Vert\, P_{Y'|X}\circ P_{X\mid M})  \right] \notag \\
	&= \frac{1}{2}\biggl[\mathbb{E}_{P_M}\Bigl[\bigl\lVert(T\hx-\hy)M\bigr\rVert_{\sigt + TT^\T}^2\Bigr] + \operatorname{Tr}\left\{(\sigt + TT^\T)^{-1}\right\} + \log \mathrm{det}(\sigt + TT^\T) - d_Y\biggr],
\end{align}
where we have used the fact that $\det(I) = 1$ and $n$ here is $d_Y$.

\section{Proof of Theorem~\ref{thm_gerdes}}
\label{app_thm_gerdes}

\noindent The proof of Theorem~\ref{thm_gerdes} consists of two steps:
\begin{enumerate}
	\item First, we show that the concept of Blackwell sufficiency (Definition~\ref{def_blackwell}) is identical to another concept known as stochastic degradedness (defined formally below) in the literature on broadcast channels.
		We show this in Lemma~\ref{lem_bs_sd_equiv}.
	\item Then, we leverage previous work from the literature characterizing the stochastic degradedness of Gaussian MIMO broadcast channels.
\end{enumerate}
Finally, we provide an explanation of what is expressed in footnote~\ref{fn_assm_star} in a subsection at the end of this section of the Appendix.

\subsection{The Equivalence of Blackwell Sufficiency and Stochastic Degradedness}

\begin{definition}[Stochastic degradedness]
	We say that a channel $P_{Y|M}$ is \emph{stochastically degraded} with respect to another channel $P_{X|M}$ if there exists a random variable $X'$ with sample space $\sampx$ such that $P_{X'|M} = P_{X|M}$ and $M$---$X'$---$Y$ is a Markov chain.
\end{definition}
\vspace{6pt}
\begin{lemma}[Equivalence of Blackwell sufficiency and stochastic degradedness] \label{lem_bs_sd_equiv}
	$P_{X|M}$ is Blackwell sufficient for $P_{Y|M}$ if and only if $P_{Y|M}$ is stochastically degraded with respect to $P_{X|M}$.
\end{lemma}
\vspace{-6pt}
\begin{proof}[\indent Proof]%
    Throughout this proof, we drop the arguments of probability distributions.
	When we equate two distributions, we mean that they are identical at all points in their shared domain.%
	\footnote{This can be extended to be more measure-theoretically accurate, to mean that they are equal $\mu$-almost everywhere under some measure $\mu$ that is absolutely continuous with respect to the two distributions being equated.}

	$(\Leftarrow)$ Suppose $P_{Y|M}$ is stochastically degraded w.r.t.\ $P_{X|M}$.
	Then, $\exists\; X'$ such that
	\begin{equation} \label{eq_stoch_mc}
		P_{X'|M} = P_{X|M} \text{ and } P_{YX'|M} = P_{Y|X'} P_{X'|M}.
	\end{equation}
	Let $P_{Y'|X}(y \given x) \coloneqq P_{Y|X'}(y \given x) \;\forall\; x \in \sampx$, $y \in \sampy$.
	Then, \eqref{eq_stoch_mc} implies that
	\begin{align}
		P_{Y'|X} P_{X|M} &= P_{Y|X'} P_{X'|M}.
		\shortintertext{Therefore,}
		\int P_{Y'|X} P_{X|M} \, dx &= \int P_{Y|X'} P_{X'|M} \, dx \\
									&= \int P_{YX'|M} \, dx \;=\; P_{Y|M},
	\end{align}
	which proves that $P_{X|M}$ is Blackwell sufficient for $P_{Y|M}$.

	$(\Rightarrow)$ Suppose $P_{X|M}$ is Blackwell sufficient for $P_{Y|M}$.
	Then, $\exists\; P_{Y'|X}$ such that
	\begin{equation}
		\int P_{Y'|X} P_{X|M} \, dx = P_{Y|M}.
	\end{equation}
	In other words,%
	\begin{equation} \label{eq_blk_channel}
		P_{Y'|M} = P_{Y|M} \;\text{ and hence }\; P_{Y'M} = P_{YM}.
	\end{equation}
	Let $X'$ be defined through a stochastic transformation of $Y$ and $M$: $P_{X'|YM}(x|y,m) \coloneqq P_{X|Y'M}(x|y,m) \;\forall\; x \in \sampx$, $y \in \sampy$, $m \in \sampm$.
	Then, using \eqref{eq_blk_channel}, we find
	\begin{align}
		P_{X'YM} &= P_{X'|YM} P_{YM} \\
				 &= P_{X|Y'M} P_{Y'M} = P_{XY'M}. \label{eq_prime_switch}
	\end{align}
	This in turn implies
	\begin{align}
		P_{X'Y|M} = P_{XY'|M} &\overset{(a)}{=} P_{Y'|X} P_{X|M} \\
							  &\overset{(b)}{=} P_{Y|X'} P_{X'|M},
	\end{align}
	which proves that $M$---$X'$---$Y$ is a Markov chain. In the above equation, (a) follows from how $P_{Y'|X}$ is defined, while (b) follows from~\eqref{eq_prime_switch}.
\end{proof}
\vspace{-6pt}
\begin{remark}
	The equivalence stated in Lemma~\ref{lem_bs_sd_equiv} was also mentioned in passing by Raginsky~\cite{raginsky2011shannon}, though without proof.
	On the other hand, Cover and Thomas~\cite[Sec.~15.6.2]{cover2012elements} \emph{define} stochastic degradedness similar to how we have defined Blackwell sufficiency.
	However, since we refer to the work of Gerdes et al.~\cite{gerdes2015equivalence} in what follows, we have used their definition of stochastic degradedness, and proved the equivalence formally for the sake of completeness.
\end{remark}

\subsection{Characterizing the Stochastic Degradedness of Broadcast Channels}
\label{app_sd_of_broadcast_channels}

Before proceeding to the proof of Theorem~\ref{thm_gerdes}, we introduce a lemma borrowed from Shang and Poor~\cite[Lemma~5]{shang2012noisy} that provides an equivalent characterization of the condition in Theorem~\ref{thm_gerdes}.
\begin{lemma} \label{lem_shang_poor} The condition $\hx^\T \sigx^{-1} \hx \suff \hy^\T \sigy^{-1} \hy$ holds if and only if
	\begin{equation}
		\exists\; T \quad\text{s.t.}\quad \hy = T\hx \;\text{ and }\; T \sigx T^\T \ffus \sigy.
	\end{equation}
\end{lemma}
\vspace{-6pt}
\begin{proof}[\indent Proof]
    Consider the whitened form of the channels:
    \begin{equation}
        \widetilde H_X \coloneqq \sigx^{-\frac{1}{2}}\hx,\, \widetilde H_Y \coloneqq \sigy^{-\frac{1}{2}}\hy,
    \end{equation}
	as mentioned in Remark~\ref{rem_whitening}.
    Then, we need to show $\widetilde H_X^\T \widetilde H_X \suff \widetilde H_Y^\T \widetilde H_Y$ if and only if
    \begin{equation}
        \exists\; T \quad \text{s.t.} \quad \widetilde H_Y = T \widetilde H_X \quad \text{and} \quad TT^\T \ffus I.
    \end{equation}
    The remainder of the proof follows from~\cite[Lemma~5]{shang2012noisy}.
\end{proof}
\begin{proof}[\indent Proof of Theorem~\ref{thm_gerdes}]
	The statement of Theorem~\ref{thm_gerdes} specifies a necessary and sufficient condition for $X \suff_M Y$ for multivariate Gaussian $P_{MXY}$.
	However, thanks to Lemma~\ref{lem_bs_sd_equiv}, the same condition is necessary and sufficient for $P_{Y|M}$ being stochastically degraded with respect to $P_{X|M}$, and so this is what we proceed to show.
	This proof is derived in large part from the work of Gerdes et al.~\cite{gerdes2015equivalence}.

	$(\Leftarrow)$ Suppose that $\hx^\T \sigx^{-1} \hx \suff \hy^\T \sigy^{-1} \hy$. Then, by Lemma~\ref{lem_shang_poor}, $\exists\; T$ such that $\hy = T\hx$ and $T \sigx T^\T \ffus \sigy$. In other words, we may write
	\begin{equation}
		Y' = TX + N,
	\end{equation}
	where $N \sim \mathcal N(0, \sigt)$, with $\sigt \coloneqq \sigy - T \sigx T^\T$ (the fact that $T \sigx T^\T \ffus \sigy$ ensures that $\sigt$ is positive semidefinite, and hence a valid covariance matrix). Therefore, $\exists\; Y'$ generated by a stochastic transformation $P_{Y'|X} = \mathcal N(TX,\, \sigt)$, such that
	\begin{equation}
		P_{Y'|M} = \mathcal N(T \hx M,\, T \sigx T^\T + \sigt)
				 = \mathcal N(\hy M,\, \sigy) = P_{Y|M}.
	\end{equation}
	Hence, by Definition~\ref{def_blackwell}, $X \suff_M Y$.

	$(\Rightarrow)$ Next, suppose that $X \suff_M Y$. Then, $\exists\; Y'$ generated by $P_{Y'|X}$, such that $P_{Y'|M} = P_{Y|M}$. Since $M$---$X$---$Y'$ is a Markov chain,
	\begin{equation} \label{eq_dpi}
		I(M; X) \overset{(a)}{\geq} I(M; Y') \overset{(b)}{=} I(M; Y),
	\end{equation}
	where (a) follows from the Data Processing Inequality~\cite[Ch.~2]{cover2012elements}, and (b) follows because $P_{Y'M} = P_{YM}$. For a Gaussian channel $P_{X|M}$, the mutual information is given by~\cite[Thm.~8.4.1]{cover2012elements}
	\begin{align}
		I(M; X) &= h(X) - h(X \given M) \\
				&= \frac{1}{2} \log\det\bigl( (2 \pi e)(\sigx + \hx \sigm \hx^\T) \bigr)  - \frac{1}{2} \log\det( 2 \pi e \sigx ) \\
				&= \frac{1}{2} \log\biggl(\frac{\det( \sigx + \hx \sigm \hx^\T )}{\det( \sigx )}\biggr) \\
				&= \frac{1}{2} \log\biggl(\frac{\det( \sigx ) \det( I + \sigx^{-1} \hx \sigm \hx^\T )}{\det( \sigx )}\biggr) \\
				&= \frac{1}{2} \log\det( I + \sigx^{-1} \hx \sigm \hx^\T ).
	\end{align}
	Now, suppose for the sake of contradiction that $\hx^\T \sigx^{-1} \hx \nsuff \hy^\T \sigy^{-1} \hy$.
	Then, by the definition of positive semidefiniteness, $\exists\; c \in \mathbb R^{d_M}$ such that
	\begin{equation}
		c^\T \hx^\T \sigx^{-1} \hx c < c^\T \hy^\T \sigy^{-1} \hy c.
	\end{equation}
	Since $\log\det( I + AB ) = \log\det( I + BA )$, $\log(\cdot)$ is an increasing function, and the determinant of a scalar is equal to itself, we have that
	\begin{align}
		\hspace{-2mm}\log\det( 1 + c^\T \hx^\T \sigx^{-1} \hx c ) &< \log\det( 1 + c^\T \hy^\T \sigy^{-1} \hy c ) \\
		\hspace{-2mm}\Rightarrow \log\det( I + \sigx^{-1} \hx c c^\T \hx^\T ) &< \log\det( I + \sigy^{-1} \hy c c^\T \hy^\T )
	\end{align}
	Now, since $cc^\T \suff 0$, it is a valid covariance matrix. So if we set $\sigm \coloneqq cc^\T$, we get
	\begin{align}
		I(M ; X) &= \log\det( I + \sigx^{-1} \hx \sigm \hx^\T ) \\
				 &< \log\det( I + \sigy^{-1} \hy \sigm \hy^\T ) \\
				 &= I(M ; Y).
	\end{align}
	However, this contradicts~\eqref{eq_dpi}, which holds no matter what $\sigm$ is. Therefore, we must have $\hx^\T \sigx^{-1} \hx \suff \hy^\T \sigy^{-1} \hy$.
\end{proof}
\vspace{-6pt}

\subsection{Explanation of Footnote~\ref{fn_assm_star}: The Difference Between PIDs Satisfying Assumption $(*)$ and Blackwellian PIDs}

In footnote~\ref{fn_assm_star}, we noted that Barrett's result applied to all PIDs satisfying Assumption~\eqref{eq_assm_star} from Bertschinger et al.~\citep{bertschinger2014quantifying}.
On the other hand, we consider all Blackwellian PIDs when showing the extension of Barrett's result.
Here, we discuss the commonalities and differences between these two sets of PID definitions.

First, we formally restate Assumption~\eqref{eq_assm_star} as given in Bertschinger et al.~\citep{bertschinger2014quantifying}:
\begin{equation} \label{eq_assm_star}
	\ui{X}{Y} \text{ is a function of only } P_M, P_{X|M} \text{ and } P_{Y|M}, \hfill \tag{$*$}
\end{equation}
as opposed to being a function of the complete joint distribution $P_{MXY}$.
Barrett's main result showed that for jointly Gaussian $P_{MXY}$ with scalar $M$, all PIDs that satisfied the above Assumption~\eqref{eq_assm_star} reduced to the MMI-PID.
However, we show a seemingly different result: Theorem~\ref{thm_gerdes} states that for fully multivariate Gaussian $P_{MXY}$ satisfying equation~\eqref{eq_gauss_bs_cond}, all \emph{Blackwellian} PIDs reduce to the MMI-PID.
A natural question that arises is: what is the relationship between Blackwellian PIDs and PIDs that satisfy Assumption~\eqref{eq_assm_star}?

In fact, PIDs that satisfy Assumption~\eqref{eq_assm_star} and PIDs that are Blackwellian form two distinct and unrelated sets.
All Blackwellian PIDs need not satisfy Assumption~\eqref{eq_assm_star}, and all PIDs that satisfy Assumption~\eqref{eq_assm_star} need not be Blackwellian.

By virtue of how it is defined, $\widetilde{UI}$ upper bounds all PIDs that satisfy Assumption~\eqref{eq_assm_star}~(see \cite[Lem.~3]{bertschinger2014quantifying}).
Thus, if $UI^{(*)}$ is an arbitrary PID that satisfies Assumption~\eqref{eq_assm_star}, then whenever $\widetilde{UI}_X$ goes to zero, $UI^{(*)}_X$ must also go to zero.
Crucially, however, the converse does not hold: $\widetilde{UI}_X > 0$ does \emph{not} imply that $UI^{(*)}_X$ must also be \emph{non}-zero.%
\footnote{A simple example might be the MMI-PID, which satisfies Assumption~\eqref{eq_assm_star}, and for which at least one of $UI_X$ or $UI_Y$ must be zero.
However, there are always instances when $\widetilde{UI}_X > 0$ and $\widetilde{UI}_Y > 0$, such as the counterexample shown in Appendix~\ref{app_counterexample}.}
This feature is what distinguishes PIDs satisfying Assumption~\eqref{eq_assm_star} from Blackwellian PIDs, as explained below.

By definition, all Blackwellian PIDs have $UI_X$ going to zero \emph{if and only if} $Y \suff_M X$, and thus they all have $UI_X$ going to zero \emph{together}.
Since $\widetilde{UI}$ is Blackwellian~\citep{bertschinger2014quantifying}, for any arbitrary Blackwellian PID $UI^{\text{Bw}}$, whenever $\widetilde{UI}$ goes to zero, $UI^{\text{Bw}}$ also goes to zero, and further, whenever $\widetilde{UI}$ is non-zero, $UI^{\text{Bw}}$ is also non-zero.

We can summarize these relationships as:
\begin{equation}
	UI^\text{Bw}_X = 0 \quad \Leftrightarrow \quad \widetilde{UI}_X = 0 \quad \Rightarrow \quad UI^{(*)}_X = 0
\end{equation}
Therefore, our extension of Barrett's result (Corollary~\ref{cor_gauss_pid}) can also be stated for PIDs satisfying Assumption~\eqref{eq_assm_star} rather than Blackwellian PIDs (since Corollary~\ref{cor_gauss_pid} is effectively a one-way implication).
However, while Theorem~\ref{thm_gerdes} can also be written (for jointly Gaussian $P_{MXY}$) as
\begin{equation}
	UI^{\text{Bw}}_Y = 0 \quad \Leftrightarrow \quad \hx^\T \sigx^{-1} \hx \suff \hy^\T \sigy^{-1} \hy,
\end{equation}
we cannot write the same with $UI^{(*)}_Y = 0$ in place of $UI^{\text{Bw}}_Y = 0$, since the forward implication ($\Rightarrow$) would not always hold.

\section{Justification for Our Convex Approximation Formulation}
\label{app_justification}

First, recall the expression for Gaussian deficiency from equation~\eqref{eq_gaussdef2}, which we are trying to approximate through a convex objective:
\begin{equation} \label{eq_gaussdef_app}
    \dfg{Y}{X} = \inf_{T,\sigt\suff 0} \frac{1}{2}\biggl[\begin{aligned}[t]
			&\ \mathbb{E}_{P_M}\Bigl[\bigl\lVert(T\hx-\hy)M\bigr\rVert_{\sigt + TT^\T}^2\Bigr] \\
			&+ \operatorname{Tr}\left\{(\sigt + TT^\T)^{-1}\right\}
			+ \log\det(\sigt + TT^\T) - d_Y
		\biggr]\end{aligned}
\end{equation}
As stated in the main text, the derivation of our proposed convex approximation has three main steps:
\begin{enumerate}
    \item First, we obtain a condition on $\sigt$ (in terms of $T$), which is locally optimal. This condition allows us to reduce the optimization problem from two variables ($\sigt$ and $T$) to just one ($T$~alone). This step is \emph{exact}, and is shown in Proposition~\ref{prop_gaussdef_simplifcn}.
    \item Next, we reinterpret and approximate the objective to significantly simplify its functional form, while attempting to minimize the same entities.
    This objective is designed to specifically to ensure that the approximate deficiency recovers the true Gaussian deficiency at its extremal values (i.e. when $\delta_G$ is equal to zero, or equal to the mutual information).
    \item Lastly, as part of reducing the optimization problem to a single variable, the constraint $\sigt \suff 0$ is now replaced by a positive semidefiniteness constraint in terms of $T$. We approximate and simplify this constraint to make it more amenable for a convex program to handle.
\end{enumerate}
We first state Proposition~\ref{prop_gaussdef_simplifcn}, which encapsulates segments of the simplification process that are \emph{exact}. The parts of the simplification that actually constitute approximations are described immediately afterward.

\begin{proposition} \label{prop_gaussdef_simplifcn}
    The Gaussian deficiency, as given by equation \eqref{eq_gaussdef_app}, finds a local minimum at
    \begin{equation}
        \sigt = I - TT^\T + (\hy - T\hx) \sigm (\hy - T\hx)^\T.
    \end{equation}
    Substituting this expression back into the optimization problem significantly simplifies the expression for the Gaussian deficiency:
    \begin{equation}
		\begin{gathered}
            \dfg{Y}{X} \; = \; \inf_{T} \; \log\det( I + (\hy - T\hx) \sigm (\hy - T\hx)^\T ) \\
            \textup{s.t.} \quad I - TT^\T + (\hy - T\hx) \sigm (\hy - T\hx)^\T \suff 0
        \end{gathered}
    \end{equation}
\end{proposition}%
\vspace{-6pt}%
\begin{proof}[\indent Proof]
The expression in equation~\eqref{eq_gaussdef_app} can be simplified into a more manageable form, starting with the expectation:
\begin{align}
    \MoveEqLeft \mathbb{E}_{P_M}\Bigl[ \bigl\lVert (\hy - T\hx) M \bigr\rVert_{\sigt + TT^\T}^2 \Bigr] \notag \\
    &\overset{(a)}{=} \mathbb{E}_{P_M}\Bigl[ M^\T (\hy - T\hx)^\T (\sigt + TT^\T)^{-1} (\hy - T\hx) M \Bigr] \\
    &\overset{(b)}{=} \mathbb{E}_{P_M}\Bigl[ \operatorname{Tr} \bigl\{ M^\T (\hy - T\hx)^\T (\sigt + TT^\T)^{-1} (\hy - T\hx) M \bigr\}\Bigr] \\
    &\overset{(c)}{=} \mathbb{E}_{P_M}\Bigl[ \operatorname{Tr}\bigl\{ (\sigt + TT^\T)^{-1} (\hy - T\hx) M M^\T (\hy - T\hx)^\T \bigr\}\Bigr] \\
    &\overset{(d)}{=} \operatorname{Tr}\Bigl\{ (\sigt + TT^\T)^{-1} (\hy - T\hx) \mathbb{E}_{P_M} \bigl[M M^\T\bigr] (\hy - T\hx)^\T \Bigr\} \\
    &\overset{(e)}{=} \operatorname{Tr}\Bigl\{ (\sigt + TT^\T)^{-1} (\hy - T\hx) \sigm (\hy - T\hx)^\T \Bigr\},
\end{align}
where in (a), we have simply expanded the Mahalanobis norm; in (b), we have used the fact that a scalar is equal to the trace of itself; in (c), we have relied on the property that the trace is invariant under cyclic permutations; in (d), we have moved the expectation inside, which is possible since all these operations are linear; and finally in (e), we have equated $\mathbb E[M M^\T]$ to $\sigm$.

Now that we have resolved the expectation, we can substitute this expression back into \eqref{eq_gaussdef_app}, and simplify the expression further:
\begin{align}
    \dfg{Y}{X} &= \inf_{T,\sigt\suff 0} \frac{1}{2}\biggl[\begin{aligned}[t]
			&{} \operatorname{Tr}\Bigl\{ (\sigt + TT^\T)^{-1} (\hy - T\hx) \sigm (\hy - T\hx)^\T \Bigr\} \\
			&+ \operatorname{Tr}\left\{(\sigt + TT^\T)^{-1}\right\}
			+ \log\det(\sigt + TT^\T) - d_Y
		\biggr]\end{aligned} \\
		&= \inf_{T,\sigt\suff 0} \frac{1}{2}\biggl[\begin{aligned}[t]
			&{} \operatorname{Tr}\Bigl\{ (\sigt + TT^\T)^{-1} \bigl(I + (\hy - T\hx) \sigm (\hy - T\hx)^\T \bigr) \Bigr\} \\
			&- d_Y + \log\det(\sigt + TT^\T)
		\biggr]\end{aligned} \label{eq_trace_dy_cancel}
\end{align}
From this, it might already be intuitively apparent that we would like
\begin{gather}
    (\sigt + TT^\T)^{-1} \bigl(I + (\hy - T\hx) \sigm (\hy - T\hx)^\T \bigr) = I \\
    \Rightarrow\quad \sigt = I - TT^\T + (\hy - T\hx) \sigm (\hy - T\hx)^\T, \label{eq_sigt_ideal}
\end{gather}
since this would allow the trace term to perfectly cancel $d_Y$ in \eqref{eq_trace_dy_cancel} (as the matrix within the trace has dimension $d_Y \times d_Y$).
However, it is possible to show that \eqref{eq_sigt_ideal} in fact gives a local optimum for $\sigt$.
This can be done by examining the partial derivative of the objective in \eqref{eq_trace_dy_cancel} with respect to $\sigt$, which is what we proceed to do next.%
\footnote{The matrix derivatives in this section were computed by hand (as shown in the body of the text), but they were also checked against an online symbolic matrix differentiation tool, \url{https://www.matrixcalculus.org/}~\cite{laue2018computing, laue2020simple}.}

First, we lay out the following identities for matrices $A$ and $B$ (refer \citep[Sec.~2,~2.4~and~9.7]{petersen2012matrix}):
\begin{equation} \label{eq_matrix_deriv_identities}
	\begin{alignedat}{3}
        \partial(AB) &= (\partial A)B + A(\partial B)
            &\qquad \partial \operatorname{Tr}\{A\} &= \operatorname{Tr}\{\partial A\}
			&\qquad \frac{\partial A}{\partial A_{ij}} &= J^{ij} \\
        \partial (A^{-1}) &= - A^{-1} (\partial A) A^{-1}
            &\qquad \partial (\log \det A) &= \operatorname{Tr}\{ A^{-1} \partial A \}
			&\qquad \Tr\{AJ^{ij}\} &= A^\T\vert_{ij} = A_{ji}
    \end{alignedat}
\end{equation}
where $J^{ij}$ is the single-entry matrix, which contains $0$'s everywhere except at the $(i, j)$-th location, where it contains a $1$.
These identities apply only when taking derivatives with respect to matrices that have no special structure (e.g., symmetry).
Since $\sigt$ is a symmetric matrix, we make use of some additional results~\citep[Sec.~2.8.2]{petersen2012matrix}.
For a symmetric matrix $S$ and some scalar function $f$ that depends on $S$,
\begin{equation} \label{eq_symm_matrix_deriv}
	\frac{\partial f}{\partial S}\bigg\vert_S = \frac{\partial' f}{\partial' S} + \frac{\partial' f}{\partial' S}^\T - \frac{\partial' f}{\partial' S} \odot I
\end{equation}
where $\partial'/\partial' S$ is used to denote that the partial derivative is taken \emph{as if} $S$ had no structure (so the identities in \eqref{eq_matrix_deriv_identities} may be applied), and $\odot$ represents the element-wise multiplication of matrices (also called a Hadamard product).

Using \eqref{eq_matrix_deriv_identities} and \eqref{eq_symm_matrix_deriv}, we can write out the partial derivative of the objective in \eqref{eq_trace_dy_cancel} with respect to $\sigt$, and equate it to zero.
To simplify notation, let $A \coloneqq \sigt + TT^\T$ and $B \coloneqq I + (\hy - T\hx) \sigm (\hy - T\hx)^\T$, so that the objective in \eqref{eq_trace_dy_cancel} can be written as:
\begin{equation}
	f(\sigt) \coloneqq \frac{1}{2} \Bigl[ \Tr\{A^{-1} B \} - d_Y + \log\det(A) \Bigr]
\end{equation}
We begin by taking a partial derivative of $f$ with respect to the $(i, j)$-th element of $\sigt$, which we denote $\sigt\vert_{ij}$:
\begin{align}
	\frac{\partial' f}{\partial' \sigt\vert_{ij}}\bigg\vert_{\sigt}
	&\overset{(a)}{=} \frac{1}{2} \biggl[ \frac{\partial'}{\partial' \sigt\vert_{ij}} \Tr\{A^{-1} B \} + \frac{\partial'}{\partial' \sigt\vert_{ij}} \log\det(A) \biggr] \bigg\vert_{\sigt} \\
	&\overset{(b)}{=} \frac{1}{2} \biggl[ \Tr\Bigl\{ \frac{\partial' A^{-1}}{\partial' \sigt\vert_{ij}} B \Bigr\} + \Tr\Bigl\{A^{-1}\frac{\partial' A}{\partial' \sigt\vert_{ij}}\Bigr\} \biggr] \bigg\vert_{\sigt} \\
	&\overset{(c)}{=} \frac{1}{2} \biggl[ \Tr\Bigl\{ - A^{-1} \frac{\partial' A}{\partial' \sigt\vert_{ij}} A^{-1} B \Bigr\} + \Tr\Bigl\{A^{-1}\frac{\partial' A}{\partial' \sigt\vert_{ij}}\Bigr\} \biggr] \bigg\vert_{\sigt} \\
	&\overset{(d)}{=} \frac{1}{2} \biggl[ \Tr\Bigl\{ - A^{-1} \frac{\partial' \sigt}{\partial' \sigt\vert_{ij}} A^{-1} B \Bigr\} + \Tr\Bigl\{A^{-1}\frac{\partial' \sigt}{\partial' \sigt\vert_{ij}}\Bigr\} \biggr] \bigg\vert_{\sigt} \\
	&\overset{(e)}{=} \frac{1}{2} \Bigl[ \Tr\bigl\{ - A^{-1} J^{ij} A^{-1} B \bigr\} + \Tr\bigl\{ A^{-1} J^{ij} \bigr\} \Bigr] \\
	&\overset{(f)}{=} \frac{1}{2} \Bigl[ \Tr\bigl\{ - A^{-1} B A^{-1} J^{ij} \bigr\} + \Tr\bigl\{ A^{-1} J^{ij} \bigr\} \Bigr] \\
	&\overset{(g)}{=} \frac{1}{2} \Bigl[ - \bigl( A^{-1} B A^{-1} \bigr)^\T_{ij} + \bigl( A^{-1} \bigr)^\T_{ij} \Bigr] \\
	&\overset{(h)}{=} \frac{1}{2} \bigl( - A^{-1} B A^{-1} + A^{-1} \bigr)_{ij} \\
	\Rightarrow \quad \frac{\partial' f}{\partial' \sigt}\bigg\vert_{\sigt}
	&\overset{(i)}{=} \frac{1}{2} \bigl( - A^{-1} B A^{-1} + A^{-1} \bigr),
\end{align}
where in the above labelled steps, we have:
\begin{enumerate}[label=(\alph*)]
	\item Taken the partial derivative of $f$ w.r.t.~$\sigt\vert_{ij}$ \emph{as if} $\sigt$ had no special structure, and evaluated it at $\sigt$
	\item Used the identities in \eqref{eq_matrix_deriv_identities}, noting that $B$ is constant w.r.t.~$\sigt$
	\item Used the identities in \eqref{eq_matrix_deriv_identities}
	\item Used the fact that $A = \sigt + TT^\T$, and that the partial derivative of $T$ with respect to $\sigt\vert_{ij}$ is zero
	\item Used the identities in \eqref{eq_matrix_deriv_identities}
	\item Made use of the fact that the trace is invariant under cyclic permutations of its argument
	\item Used the identities in \eqref{eq_matrix_deriv_identities}
	\item Relied on the fact that $A$ and $B$ are symmetric, so that $(A^{-1} B A^{-1})^\T = A^{-1} B A^{-1}$, and $A^{-\T} = A^{-1}$
	\item Assembled the full partial derivative with respect to the \emph{matrix} $\sigt$
\end{enumerate}
Hence, we can write out the complete partial derivative of the objective in \eqref{eq_trace_dy_cancel} with respect to $\sigt$, and set it to zero:
\begin{align}
	0 = \frac{\partial f}{\partial \sigt}\bigg\vert_{\sigt}
	&= \biggl[ \frac{\partial' f}{\partial' \sigt} + \frac{\partial' f}{\partial' \sigt}^\T - \frac{\partial' f}{\partial' \sigt} \odot I \biggr] \bigg\vert_{\sigt} \\
	&= \frac{1}{2} \Bigl[ \bigl( - A^{-1} B A^{-1} + A^{-1} \bigr) + \bigl( - A^{-1} B A^{-1} + A^{-1} \bigr)^\T - \bigl( - A^{-1} B A^{-1} + A^{-1} \bigr) \odot I \Bigr] \\
	&= \bigl( - A^{-1} B A^{-1} + A^{-1} \bigr) - \frac{1}{2} \bigl( - A^{-1} B A^{-1} + A^{-1} \bigr) \odot I,
\end{align}
where in the last step, we have again relied on the symmetry of $A$ and $B$.

A sufficient condition for the above equation to hold is
\begin{alignat}{2}
	& &\quad - A^{-1} B A^{-1} + A^{-1} &= 0 \\
	&\Rightarrow &\quad A^{-1} B A^{-1} &= A^{-1} \\
	&\Rightarrow &\quad B A^{-1} &= I \\
	&\Rightarrow &\quad B &= A
\end{alignat}
Substituting the original expressions for $A$ and $B$,
\begin{align}
    \sigt + TT^\T &= I + (\hy - T\hx) \sigm (\hy - T\hx)^\T \\
    \Rightarrow\qquad\qquad \sigt &= I - TT^\T + (\hy - T\hx) \sigm (\hy - T\hx)^\T
\end{align}
Finally, we can substitute this expression back into the objective in \eqref{eq_trace_dy_cancel} to cancel the $\operatorname{Tr}\{\cdot\}$ and $d_Y$ terms, leaving only the $\log\det(\cdot)$ term:
\begin{equation}
	\begin{gathered}
		\dfg{Y}{X} \; = \; \inf_{T} \; \log\det(\sigt + TT^\T) \\
        \textup{s.t.} \quad I - TT^\T + (\hy - T\hx) \sigm (\hy - T\hx)^\T \suff 0
    \end{gathered}
\end{equation}
Substituting for the $\sigt$ within the $\log\det(\cdot)$ yields the desired expression from the statement of the proposition.
\end{proof}

This concludes the first step of simplifying our optimization problem, which reduces the problem to an optimization over a single variable, and which is exact. Next, we proceed to approximate the terms in this simplified formulation.

First, note that the objective seeks to minimize the log-determinant of $I + (\hy - T\hx) \sigm (\hy - T\hx)^\T$. This term would naturally become smaller if $\hy$ is made closer to $T\hx$. To simplify this objective, we instead look at minimizing
\begin{equation}\label{eq_obj_orig}
    \mathbb{E}_{P_M}\Bigl[ \bigl\lVert (\hy - T\hx) M \bigr\rVert_{\sigt + TT^\T}^2 \Bigr],
\end{equation}
which appears in the original optimization problem. However, to avoid having the Mahalanobis norm being taken against $\sigt + TT^\T$, which depends on $T$ and yields a non-trivial objective function, we instead minimize:
\begin{equation}\label{eq_obj_relax}
    \mathbb{E}_{P_M}\Bigl[ \bigl\lVert (\hy - T\hx) M \bigr\rVert_{I + \hy\sigm\hy^\T}^2 \Bigr].
\end{equation}
This objective function is chosen such that \eqref{eq_obj_orig} and \eqref{eq_obj_relax} will be equivalent ``at the extremes,'' i.e. when the deficiency is zero ($H_Y-TH_X=0$) and when then deficiency is the mutual information ($T=0$). This choice also enables the proof of Proposition~\ref{prop_extremes} (see Appendix~\ref{app_prop_extremes}), which guarantees that when the true deficiency equals zero or the mutual information, the approximation yields the true deficiency. We use this revised objective to estimate the optimal \emph{argument} $\widehat{T}$. To compute the approximate Gaussian deficiency, we substitute this value of $\widehat{T}$ back into the formula for the Gaussian deficiency, given by \eqref{eq_gaussdef2}.

Finally, to arrive at the estimate of $\sigth$ we employ in our approximation, given in equation~\eqref{eq_sigth}, we assume that $\hy - T\hx$ is small in one part of the equation:
\begin{align}
    \sigt &= I - TT^\T + (\hy - T\hx) \sigm (\hy - T\hx)^\T \\
    &= I - TT^\T + \hy\sigm\hy^\T - \hy\sigm\hx^\T T^\T - T\hx\sigm\hy^\T + T\hx\sigm\hx^\T T^\T \\
    &\approx I - TT^\T + \hy\sigm\hy^\T - T\hx\sigm\hx^\T T^\T - T\hx\sigm\hx^\T T^\T + T\hx\sigm\hx^\T T^\T \label{eq_sigth_approx} \\
    &= I - TT^\T + \hy\sigm\hy^\T - T\hx\sigm\hx^\T T^\T \\
    &= I + \hy\sigm\hy^\T - T(I + \hx\sigm\hx^\T)T^\T
\end{align}
where the approximation occurs in \eqref{eq_sigth_approx} for the two terms that are linear in $T$. Since $\sigt \suff 0$ becomes a constraint in our eventual optimization problem, this approximation significantly simplifies the constraint: firstly, it removes the linear terms making this a pure quadratic form; secondly, the matrix appearing in the quadratic atom, $I + \hx\sigm\hx^\T$, is now guaranteed to be invertible. Prior to the approximation, the same matrix was $I-\hx\sigm\hx^\T$, which did not have the same guarantee. These simplifications enable using the Schur complement in formulating our constrained optimization problem, as detailed in Appendix \ref{app_cvx}.

\section{The Approximate Deficiency is Well Defined}\label{app_lemma_valid}
\begin{lemma}\label{lemma_valid}
For jointly Gaussian random vectors $M$, $X$ and $Y$ satisfying the assumption in Remark~\ref{rem_full_rank}, $\widehat{T}$ and $\sigth$ as defined by \eqref{eq_sigth} and \eqref{eq_th} yield an approximate deficiency $\dfh{Y}{X}$ that is well defined, i.e. the covariance matrix $\sigth + \widehat{T}\widehat{T}^\T$ for the composite channel $P_{Y'|X}\circ P_{X\mid M}$ is invertible.
\end{lemma}
\begin{proof}
We know that $\sigth\suff 0$ by virtue of the constraint in \eqref{eq_th} and thus $\sigth+\widehat{T}\widehat{T}^\T\suff 0$ as well. For ease of notation, let $A \coloneqq \hy\sigm\hy^\T$ and $B \coloneqq \hx\sigm\hx^\T$. Assume for a contradiction that $\sigth+\widehat{T}\widehat{T}^\T$ is rank deficient. This implies that there exists a vector $v\ne \mathbf{0}$ such that:
\begin{alignat*}{2}
    & & v^\T(\sigth+\widehat{T}\widehat{T}^\T)v &= v^\T(I+A-\widehat{T}B\widehat{T}^\T)v = 0 \\
    \Rightarrow& \quad & v^\T(I+A)v &= v^\T(\widehat{T}B\widehat{T}^\T)v \\
    \overset{(a)}{\Rightarrow}& & (\widehat{T}^\T v)^\T B(\widehat{T}^\T v) &> 0 \\
    \Rightarrow& & \widehat{T}^\T v &\ne \mathbf{0}
\end{alignat*}
where $(a)$ follows from $A=(\hy\sigm^{-\frac{1}{2}})(\hy\sigm^{-\frac{1}{2}})^\T\Rightarrow A\suff 0$. But since $\sigth\suff 0$, we have:
\begin{alignat*}{2}
    & & v^\T\sigth v &\ge 0 \\
    \Rightarrow& \quad & v^\T(\sigth+\widehat{T}\widehat{T}^\T) v - v^\T(\widehat{T}\widehat{T}^\T) v &\ge 0 \\
    \Rightarrow& & v^\T\widehat{T}\widehat{T}^\T v &= \Vert \widehat{T}^\T v \Vert_2^2 \le 0.
\end{alignat*}
which is a contradiction since $\widehat{T}^\T v \ne \mathbf{0}\Rightarrow \Vert \widehat{T}^\T v \Vert_2^2 > 0$.
\end{proof}

\section{Proof of Proposition \ref{prop_extremes}}\label{app_prop_extremes}

\begin{proof}[Proof of Proposition~\ref{prop_extremes}]
For ease of notation, we omit the function arguments $(M:Y\setminus X)$ and refer simply to $\widehat{\delta}_G$, $\delta_G$, and $\delta$.

\vspace{6pt}
($\delta = 0 \Rightarrow \widehat{\delta}_G = 0$)
It was shown by Torgersen~\cite[Theorem~8.2.13]{torgersen1991comparison} that $\delta=0 \Rightarrow \delta_G=0$, thus it suffices to show that $\delta_G=0 \Rightarrow \widehat{\delta}_G=0$.

When $\delta_G$ is well-defined, we know that $\sigt+TT^\T\suff 0$ is invertible and thus $\sigt+TT^\T,(\sigt+TT^\T)^{-1}\succ 0$. This implies that the first term of \eqref{eq_gaussdef2} is greater than zero unless $(T\hx-\hy)M=0$. Thus $\delta_G=0$ implies that $(T\hx-\hy)M=0$ almost everywhere in $M$, i.e. that there exists a $T^*$ such that $T^*\hx=\hy$. Thus, we also know that the objective function in \eqref{eq_th} is minimized by $T=T^*$. To show that $T^*$ is in the feasible set, we note that $\delta_G=0$ implies that there exists a $\sigt^*\suff 0$ such that $\sigt^*+T^*{T^*}^\T=I\Rightarrow I-T^*{T^*}^\T \suff 0$. As such, when $T=T^*$ we have that $\sigt=I+\hy\sigm\hy^\T - T^*(I+\hx\sigm\hx^\T){T^*}^\T=I-T^*{T^*}^\T=\sigt^*$.

($\delta=0\Leftarrow\widehat{\delta}_G=0$) This follows from $\widehat{\delta}_G\ge\delta\ge 0$.

($\delta=I(M;Y)\Rightarrow\widehat{\delta}_G=I(M;Y)$) Note that if $\widehat{T}=\mathbf{0}$, we have
\begin{equation}
    \begin{array}{rc}
        & \sigth = I + \hy\sigm\hy^\T \\
        \Rightarrow & \widehat{P}_{Y'|X}\circ P_{X\mid M} = P_Y \\
        \Rightarrow & \widehat{\delta}_G = I(M;Y).
    \end{array}
\end{equation}
As such, it suffices to show that $\delta=I(M;Y)$ implies that the objective function in \eqref{eq_th} is minimized at $T=\mathbf{0}$. Suppose for a contradiction that there exists a $\widetilde{T}$ such that:
\begin{equation}
    \mathbb{E}_{P_M}\left[\Vert(\widetilde{T}\hx-\hy)M\Vert_{A}^2\right] < \mathbb{E}_{P_M}\left[\Vert\hy M\Vert_{A}^2\right] \label{eq_tildet}
\end{equation}
with $A\coloneqq I + \hy \sigm \hy^\T$ giving the marginal covariance matrix for $Y$ (i.e. $P_Y=\mathcal{N}(\mathbf{0},A)$). This implies that
\begin{equation}
    \mathbb{E}_{P_M}\left[\Vert(\lambda\widetilde{T}\hx-\hy)M\Vert_{A}^2\right] < \mathbb{E}_{P_M}\left[\Vert\hy M\Vert_{A}^2\right]
\end{equation}
for any $\lambda\in(0,1)$.  The implication follows from the convexity of the Mahalanobis distance, noting that $\lambda\widetilde{T}=\lambda\widetilde{T}+(1-\lambda)\mathbf{0}$. In other words, we can keep making $\widetilde{T}$ ``smaller'' (by multiplying it by smaller values of $\lambda$), until $\widetilde{\Sigma}_T\coloneqq A-\widetilde{T}\widetilde{T}^\T$ is PSD, while still making sure that \eqref{eq_tildet} holds. Defining $\widetilde{P}_{Y'|X}\coloneqq \mathcal{N}(\widetilde{T}X,\widetilde{\Sigma}_T)$, we note that $\widetilde{P}_{Y'|X}\circ P_{X\mid M}= \mathcal{N}(\widetilde{T}\hx M,A)$. As such:
\begin{align*}
    \MoveEqLeft[1] \mathbb{E}_{P_M}\left[D(P_{Y\mid M} \Vert \widetilde{P}_{Y'|X}\circ P_{X\mid M})\right] - I(M;Y) \\
    &= \mathbb{E}_{P_M}\left[D(P_{Y\mid M} \Vert \widetilde{P}_{Y'|X}\circ P_{X\mid M}) - D(P_{Y\mid M} \Vert P_Y)\right] \\
    &= \mathbb{E}_{P_M}\left[\Vert(\widetilde{T}\hx-\hy)M\Vert_{A}^2\right] - \mathbb{E}_{P_M}\left[\Vert\hy M\Vert_{A}^2\right]<0.
\end{align*}
But since $\widetilde{P}_{Y'|X}\in\mathcal{C}_G(\sampy\mid\sampx)$, we have:
\begin{equation}
    \hspace{-1.5mm}\delta_G \le \mathbb{E}_{P_M}\left[D(P_{Y\mid M} \Vert \widetilde{P}_{Y'|X}\circ P_{X\mid M})\right] < I(M;Y) \le \delta,
\end{equation}
which is a contradiction.
\end{proof}

\section{Implementation Details}
\label{app_cvx}

We briefly discuss how to reformulate the optimization problem in \eqref{eq_th} such that it satisfies the disciplined convex programming (DCP) rules and can be solved using the CVXPY software package \cite{diamond2016cvxpy,agrawal2018rewriting}. First we note that the objective function can be rewritten using the trace trick. Letting $A=I + \hy \sigm \hy^\T$ for ease of notation:
\begin{align*}
    \MoveEqLeft[1] \mathbb{E}_{P_M}\left[\Vert(T\hx-\hy)M\Vert_{A}^2\right]\\
    &= \mathbb{E}_{P_M}\left[ M^\T(T\hx-\hy)^\T A^{-1}(T\hx-\hy)M \right]\\
    &= \mathbb{E}_{P_M}\left[\operatorname{Tr}\left( M^\T(T\hx-\hy)^\T A^{-1}(T\hx-\hy)M\right) \right]\\
    &= \mathbb{E}_{P_M}\left[\operatorname{Tr}\left( MM^\T(T\hx-\hy)^\T A^{-1}(T\hx-\hy)\right) \right]\\
    &= \operatorname{Tr}\left( \sigm(T\hx-\hy)^\T A^{-1}(T\hx-\hy)\right)\\
    &= \Vert A^{-\frac{1}{2}}T\hx\sigm^{\frac{1}{2}}-A^{-\frac{1}{2}}\hy\sigm^{\frac{1}{2}} \Vert_F^2
\end{align*}
where $\Vert\cdot\Vert_F^2$ is the squared Frobenius norm. Next, the constraint can be rewritten using the Schur complement\ifshort{}{ \cite{zhang2006schur}}:
\begin{gather*}
    I + \hy\sigm\hy^\T - T(I+\hx\sigm\hx^\T)T^\T \suff 0 \\
    \Updownarrow \\
    \begin{bmatrix}
    I + \hy\sigm\hy^\T & T \\
    T^\T & (I+\hx\sigm\hx^\T)^{-1}
    \end{bmatrix}
    \suff 0
\end{gather*}
Combining the two, we obtain the problem in a form that can solved directly by CVX:
\begin{align}\label{eq_cvx_prob}
\begin{split}
\widehat{T} = \
\underset{T}{\operatorname{argmin}}\quad & \Vert A^{-\frac{1}{2}}T\hx\sigm^{\frac{1}{2}}-A^{-\frac{1}{2}}\hy\sigm^{\frac{1}{2}} \Vert_F^2 \\
\text{s.t.} \quad \quad & \begin{bmatrix}
    I + \hy\sigm\hy^\T & T \\
    T^\T & (I+\hx\sigm\hx^\T)^{-1}
    \end{bmatrix}
    \suff 0
\end{split}
\end{align}
We solved the problem used the splitting conic solver (SCS) \cite{ocpb:16,scs} with relaxation parameter 1, maximum iteration number 5000, and convergence tolerance of 1e-10. The maximum iteration number and convergence tolerance were chosen to be conservatively large and small, respectively, as we found lower and higher values of these parameters to yield inferior results (i.e. larger estimated deficiencies) and computational time was not an issue. Due to issues of numerical imprecision, the estimated $\widehat\Sigma_T$ was occasionally found to have very small negative eigenvalues (i.e. greater than -1e-6). To ensure that the estimated $\widehat\Sigma_T$ was positive semidefinite, we replaced all negative eigenvalues with zero.

\section{Step-By-Step Experimental Procedure }\label{app_exp}

\noindent Here, we describe the procedure used to generate the plots in Figure~\ref{fig_simplex}.
\begin{enumerate}
    \item Sample $d_M$, $d_X$, $d_Y$ according one of (S1)--(S4).
    \item Sample a covariance matrix $\Sigma\in \mathbb{R}^{d\times d}$ (where $d=d_M+d_X+d_Y$) from a standard Wishart distribution.
    \item Compute $I(M;(X,Y))$, $I(M;X)$, and $I(M;Y)$.
    \item Compute conditional mean and covariances and whiten to obtain $\hx$, $\hy$, and $\sigx=\sigy=I$.
    \item Estimate the deficiencies $\dfh{Y}{X}$ and $\dfh{X}{Y}$, using the the Python CVX package to solve \eqref{eq_cvx_prob}.
    \item Compute the approximate $\delta$-PID atoms.
    \item Check if $\widehat{RI}$ and $\widehat{SI}$ are non-negative.
    \item If $d_M=1$, check if either $\widehat{UI}_X\approx 0$ or $\widehat{UI}_Y\approx 0$.
\end{enumerate}

\section{Details of the Simulation Comparing the $\sim$-PID and the $\widehat\delta_G$-PID}
\label{app_poiss_vs_gauss}

\subsection{Simulation setup}

The results shown in Figure~\ref{fig_poiss_vs_gauss} are for a \emph{multivariate Poisson} distribution, which is defined below.
We used a small number of dimensions ($d_M = 2$ and $d_X = d_Y = 1$) in order to allow for computational tractability of the discrete $\sim$-PID estimator of Banerjee et al.~\cite{banerjee2018computing}.

We took $M_i \sim $ i.i.d.~Poisson$(\lambda = 2)$, for both components of $M$ (i.e, $i = 1, 2$).
We took $X = X_1 + X_2 + N_X$, where $X_i \sim $ Binomial$(M_i, w^X_i)$ for some ``weight'' $w^X_i$, and $N_X \sim$ Poisson$(\lambda_X = 1)$.
$Y$ was defined in the same way as $X$, with $Y_i$, $N_Y$, $X_i$ and $N_X$ all being independent of each other.
The results shown in Figure~\ref{fig_poiss_vs_gauss} are for $w^Y_1 = w^Y_2 = w^X_2 = 0.5$, while $w^X_1$ is varied between 0.0 and 1.0 in increments of 0.1 on the x-axis.

The discrete $\sim$-PID was estimated by truncating the Poisson distributions, while ensuring >95\%
For computing the $\widehat\delta_G$-PID, the Poisson distribution was approximated as a Gaussian distribution by considering the sample covariance matrix, estimated from $N = 10^6$ data points.
This was observed to produce a stable covariance estimate, following which channel parameters such as $\hx$, $\sigx$, etc.\ were extracted from the estimated covariance matrix.
The convex optimization problem in the definition of $\widehat\delta_G$-PID was solved using CVXPY (see Appendix~\ref{app_cvx} for details).

\subsection{Results and Limitations}

\noindent It should be noted that this experiment has several intrinsic limitations:
\begin{enumerate}
	\item The Poisson distributions had to be truncated for reasons of finiteness and computational tractability, before estimating the $\sim$-PID.
		As mentioned above, we endeavoured to ensure that the distribution was captured well, but the truncation could nevertheless skew the discrete-PID estimates.
	\item The estimation of the sample covariance matrix used to compute the $\widehat\delta_G$-PID could also contribute to error.
	\item Lastly and most importantly, the two PIDs being considered are different, i.e., the $\sim$-PID for the Poisson distribution, and the $\delta$-PID for the Gaussian approximation.
		These PIDs likely produce slightly different values: the $\delta$-PID is closer in spirit to the PID of Harder et al.~\cite{harder2013bivariate} (see~\cite{banerjee2018unique} for an explanation), which was shown to have different values from the $\sim$-PID by Bertschinger et al.~\cite{bertschinger2014quantifying}.
\end{enumerate}
These limitations suggest that we have no way of truly knowing how accurate either of the two estimates are.
However, their degree of agreement despite the use of completely different estimation techniques cannot be purely coincidental, and argues in favor of the practical applicability of our method.
In particular, these results provide evidence suggesting that:
\begin{enumerate}
	\item Our approximate PID is able to differentiate between cases of zero- and non-zero unique information.
	\item Our proposed method matches expected trends for every one of the four PID quantities.
	\item As a fraction of total mutual information, each PID value is reasonably close to its expected value, suggesting our method can be used to assess relative differences between PID quantities in different settings.
\end{enumerate}

\clearpage
\section{3-Dimensional Simplex Views}
\label{app_3d_views}

\begin{figure}[h!]
    \centering
    \includegraphics[height=0.9\textheight]{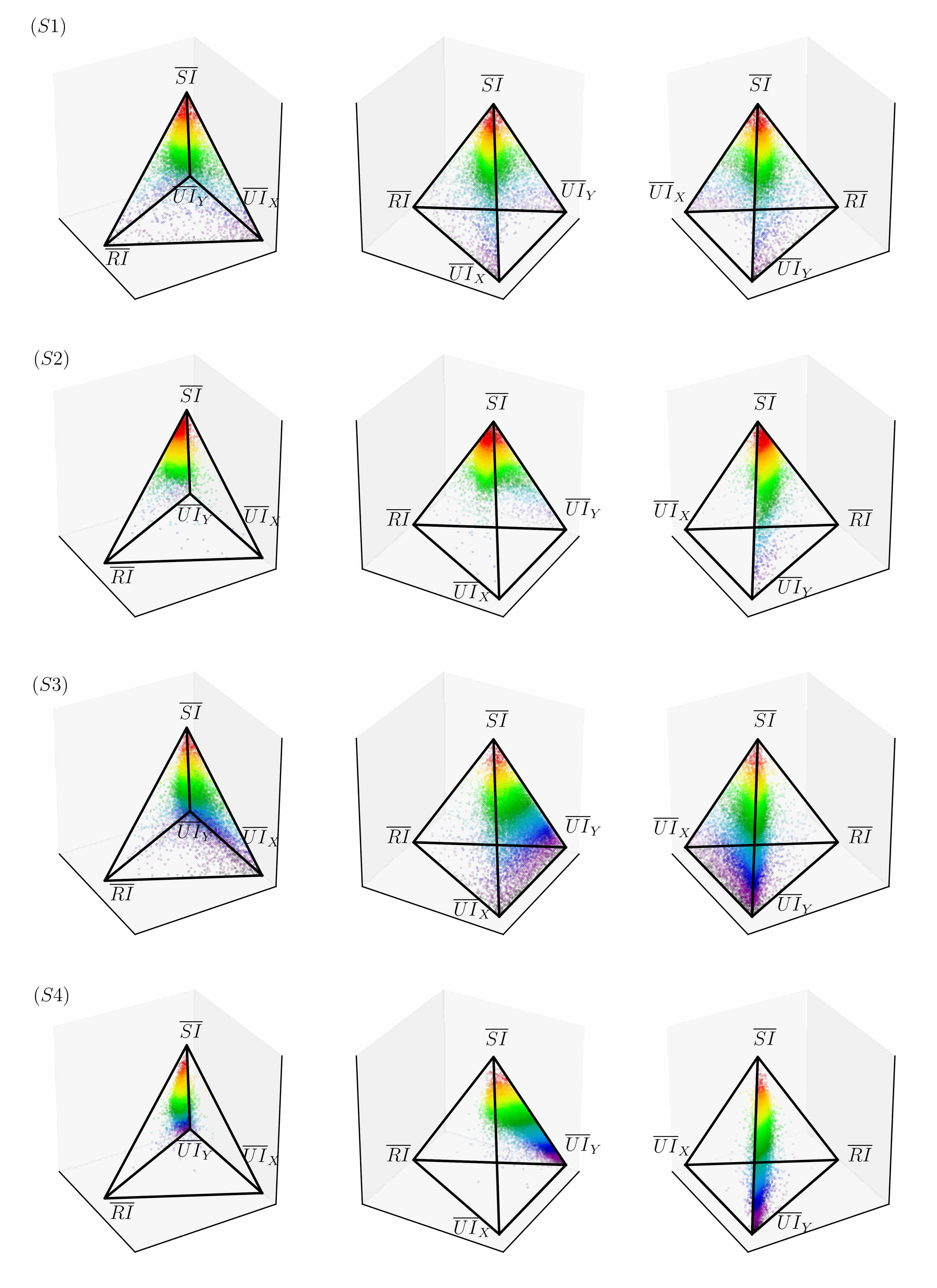}
    \caption{Alternative views of Figure \ref{fig_simplex}. Each row displays one of the four sampling schemes, and each column provides a rotated view of the simplex.}
\end{figure}

\end{document}